\documentclass[dvips]{llncs}

\makeatletter
\let\LLNCSproof\proof                         
\let\LLNCSendproof\endproof              
\let\proof\@undefined                        
\let\endproof\@undefined                  
\makeatother

\usepackage{amscd, amsmath, amssymb, amsbsy, stmaryrd, amsthm}

\let\proof\LLNCSproof                         
\let\endproof\LLNCSendproof              

\usepackage{sublabel}
\usepackage{ifthen}
\usepackage{graphicx}
\usepackage{xspace}
\usepackage{times}
\usepackage[subfigure]{tocloft}
\usepackage{subfigure}
\usepackage{float}
\usepackage{xypic}
\usepackage{vaucanson-g}
\usepackage{bbold}
\usepackage{wrapfig}

\usepackage{hyperref}

\title{\bf B\"uchi Automata can have Smaller Quotients\vspace{-4mm}}
\date{Preliminary Draft --- October 2010}
\author{Lorenzo Clemente}
\institute{LFCS. School of Informatics. University of Edinburgh. UK}

\newtheorem{lem}{Lemma}

\newtheorem{fact}{Fact}

\theoremstyle{remark}
\newtheorem*{unremark}{Remark} 

\newboolean{pedantic}
\setboolean{pedantic}{false}

\newboolean{ffalse}
\setboolean{ffalse}{false}

\newcommand{\st}{\ |\ }
\newcommand{\setof}[2]{\{ #1 \st #2 \}}

\newcommand{\restrict}[2]{\left.{#1}\right|_{#2}}
\newcommand{\card}[1]{\left|{#1}\right|}
\newcommand{\last}[1]{\mathrm{last}({#1})}

\newcommand{\eqsimrel}{\approx}

\newcommand{\lang}[2]{\mathcal L^{#1}({#2})}
\newcommand{\tracelang}[1]{\mathcal T\!r(#1)}
\newcommand{\omegatracelang}[1]{\mathcal T\!r^\omega(#1)}
\newcommand{\omegalang}[1]{\lang{\omega}{#1}}
\newcommand{\finlang}[1]{\lang{\mathrm{fin}}{#1}}

\newcommand{\kdegame}[3]{\ensuremath{G^{{#1}\textrm{-de}}({#2},{#3})}}
\newcommand{\fxsimgame}[3]{\ensuremath{G^{\textrm{de}}_{#1}({#2},{#3})}}
\newcommand{\kfxsimgame}[4]{\ensuremath{G^{#1\textrm{-de}}_{#2}({#3},{#4})}}

\newcommand{\accgame}[3]{\ensuremath{G^{#1}({#2},{#3})}}
\newcommand{\accomegagame}[2]{\ensuremath{\accgame{\omega}{#1}{#2}}}
\newcommand{\accfingame}[2]{\ensuremath{\accgame{\mathrm{fin}}{#1}{#2}}}

\newcommand{\pair}[2]{\ensuremath{\langle{#1},{#2}\rangle}}
\newcommand{\triple}[3]{\ensuremath{\langle {#1}, {#2}, {#3} \rangle}}
\newcommand{\quadruple}[4]{\langle {#1}, {#2}, {#3}, {#4} \rangle}
\newcommand{\quintuple}[5]{\langle {#1}, {#2}, {#3}, {#4}, {#5} \rangle}

\newcommand{\conf}[2]{\pair{#1}{#2}}

\newcommand{\set}[1]{\bold{#1}}

\newcommand{\quot}[2]{{#1}_{#2}}
\newcommand{\ABA}[1]{{\ensuremath{\mathcal{#1}}}}

\newcommand{\NBA}[1]{\ensuremath{\mathcal{#1}}}
\newcommand{\alttuple}{(Q, \Sigma, q_I, \Delta, E, U, F)}

\newcommand{\NBAtuple}{(Q, \Sigma, I, \Delta, F)}
\newcommand{\ABAtuple}{(A, \Sigma, \delta, \alpha)}
\newcommand{\NBAquottuple}[1]{([Q], \Sigma, [I], \Delta_{#1}, [F])}

\newcommand{\Nat}{\mathbb{N}}

\newcommand{\nextp}{\mathsf{next}}
\newcommand{\update}{\mathsf{up}}

\newcommand{\T}{\ensuremath{T}}

\newcommand{\CPrename}{\mathsf{CPre}}
\newcommand{\CPre}[2]{\CPrename^{#1}({#2})}

\newcommand{\true}{\mathtt{true}}
\newcommand{\false}{\mathtt{false}}

\newcommand{\acceptingsince}[4]{\mathsf{accepting}^{#3}_{#4}({#1}, {#2})}
\newcommand{\goodsince}[4]{\mathsf{good}^{#3}_{#4}({#1}, {#2})}
\newcommand{\rank}[1]{\mathsf{rank}_q^w({#1})}

\newcommand{\floor}[1]{\lfloor{#1}\rfloor}
\newcommand{\goesto}[1]{\stackrel{#1}{\longrightarrow}}
\newcommand{\jgoesto}[1]{\stackrel{#1}{\dashrightarrow}}
\newcommand{\goto}[1]{\stackrel{#1}{\Longrightarrow}}

\newcommand{\cont}{\subseteq}

\newcommand{\decont}{\cont^\textrm{de}}

\newcommand{\fxsim}{\sim_\textrm{fx}}

\newcommand{\fxdesim}{\fxsim^\textrm{de}}
\newcommand{\fxkdesim}[1]{\fxsim^\textrm{${#1}$-de}}

\newcommand{\wdesim}[1]{\sqsubseteq_{#1}^\textrm{de}}
\newcommand{\wkdesim}[2]{\sqsubseteq_{#1}^\textrm{${#2}$-de}}

\renewcommand{\sim}{\sqsubseteq}
\renewcommand{\simeq}{\approx}

\newcommand{\fwdisim}{\sim^{\textrm{di}}_\textrm{fw}}
\newcommand{\desim}{\sim^\textrm{de}}
\newcommand{\kdesim}[1]{\desim_{#1}}

\newcommand{\diproxysim}{\sim^{\textrm{di}}_\textrm{xy}}
\newcommand{\diproxysimeq}{\approx^{\textrm{di}}_\textrm{xy}}
\newcommand{\deproxysim}{\sim^{\textrm{de}}_\textrm{xy}}

\newcommand{\fwdisimeq}{\simeq^{\textrm{di}}_\textrm{fw}}

\newcommand{\defxsimeq}{\approx^\textrm{de}_\textrm{fx}}

\newcommand{\prefixle}{<_{\textrm{prf}}}
\newcommand{\prefixleq}{\leq_{\textrm{prf}}}

\newcommand{\countfname}{\mathsf{cnt\textrm{-}final}}
\newcommand{\countf}[2]{\countfname({#1}, {#2})}
\newcommand{\countflast}[1]{\countfname({#1})}

\newcommand{\bwsim}{\sqsubseteq_{\textrm{bw}}}
\newcommand{\bwsimeq}{\simeq_{\textrm{bw}}}

\newcommand{\bwdisim}{\bwsim^{\textrm{di}}}
\newcommand{\bwdisimrev}{(\sqsubseteq^\textrm{di}_\textrm{bw})^{-1}}
\newcommand{\bwdisimeq}{\bwsimeq^{\textrm{di}}}

\newcommand{\taudenaught}{\tau_0^\textrm{de}}
\newcommand{\taude}{\tau_1^\textrm{de}}
\newcommand{\taufxde}{\tau_1^\textrm{fx-de}}

\newcounter{myenumerate}
\newcommand{\ignore}[1]{}

\newcommand{\appendixstatement}[3]{
\paragraph{\bf {#1}~\ref{#2}.}{\em #3}
}

\newlistof{inappendix}{app}{Appendix}
\newcommand{\appendixlemma}[5][unnamedlemma]{%
\ifappendix%
\label{thesection:#2}
\def\argone{#1}
\ifthenelse{\equal{\argone}{unnamedlemma}}{\begin{lem}}{\begin{lem}[#1]}\label{#2}%
#3%
\end{lem}%
\addtocontents{app}{fulllemmastatement\roman{lem}}%
\addtocontents{app}{fulllemmaproof\roman{lem}}%
\expandafter\def\csname fulllemmastatement\roman{lem}\endcsname{%
\setcounter{lem}{\ref{#2}}%
\addtocounter{lem}{-1}%
\ifthenelse{\equal{#1}{unnamedlemma}}{\begin{lem}[in Section~\ref{thesection:#2}]}{\begin{lem}[#1, in Section~\ref{thesection:#2}]}#4\end{lem}}%
\expandafter\def\csname fulllemmaproof\roman{lem}\endcsname{%
\begin{proof}#5\end{proof}}%
\else%
\begin{lem}\label{#2}%
#3%
\end{lem}%
\begin{proof}%
#5%
\end{proof}%
\fi%
}

\let\oparagraph\paragraph
\def\paragraph#1{\vspace{-2mm}\oparagraph{#1}\vspace{0mm}}

\begin{document}

\maketitle

\vspace{-5mm}
\begin{abstract}
\noindent
We study novel simulation-like preorders for quotienting nondeterministic B\"uchi automata. 
We define \emph{fixed-word delayed simulation},
a new preorder coarser than delayed simulation.
We argue that fixed-word simulation is the coarsest forward simulation-like preorder which can be used for quotienting B\"uchi automata,
thus improving our understanding of the limits of quotienting.
Also, we show that computing fixed-word simulation is PSPACE-complete.

On the practical side,
we introduce \emph{proxy simulations},
which are novel polynomial-time computable preorders sound for quotienting.
In particular, \emph{delayed proxy simulation} induce quotients that can be smaller by an arbitrarily large factor than direct backward simulation.
We derive proxy simulations as the product of a theory of refinement transformers:
A \emph{refinement transformer} maps preorders nondecreasingly,
preserving certain properties.
We study under which general conditions refinement transformers are sound for quotienting.

\end{abstract}
\vspace{-5mm}

\section{Introduction}

B\"uchi automata minimization is an important topic in automata theory,
both for the theoretical understanding of automata over infinite words
and for practical applications.
Minimizing an automaton means reducing the number of its states as much as possible,
while preserving the recognized language.
Minimal automata need not be unique,
and their structure does not necessarily bear any resemblance to the original model;
in the realm of infinite words, 
this holds even for deterministic models.
This hints at why exact minimization has high complexity:
Indeed, minimality checking is PSPACE-hard for nondeterministic models (already over finite words \cite{ravikumar:hard:1991}),
and NP-hard for deterministic B\"uchi automata \cite{schewe:DBA:minimization}.
Moreover, even approximating the minimal model is hard \cite{gramlich:minimizing:journal:2007}.

By posing suitable restrictions on the minimization procedure,
it is nonetheless possible to trade exact minimality for efficiency.
In the approach of \emph{quotienting},
smaller automata are obtained by merging together equivalent states,
under appropriately defined equivalences.
In particular, quotienting by simulation equivalence has proven to be an effective heuristics for reducing the size of automata in cases of practical relevance.

The notion of \emph{simulation preorder} and \emph{equivalence} \cite{Milner:newbook} is a crucial tool for comparing the behaviour of systems.
It is best described via a game between two players, Duplicator and Spoiler,
where the former tries to stepwise match the moves of the latter.
But not every simulation preorder can be used for quotienting:
We call a preorder \emph{good for quotienting} (GFQ)
if the quotient automaton (w.r.t. the induced equivalence) recognizes the same language as the original automaton.
In particular,
a necessary condition for a simulation to be GFQ is to take into account the acceptance condition:
For example, in \emph{direct} simulation \cite{DHW:simulation},
Duplicator has the additional requirement to visit an accepting state whenever Spoiler does so,
while in the coarser \emph{fair} simulation \cite{fairsimulation:02},
Duplicator has to visit infinitely many accepting states if Spoiler does so.
But, while direct simulation is GFQ \cite{aziz:minimizing},
fair simulation is not \cite{fairbisimulation:00}.%
\footnote{In fact, for B\"uchi automata it is well-known that also language equivalence is not GFQ.}
This prompted the development of \emph{delayed} simulation \cite{etessami:etal:fairsimulations:05},
a GFQ preorder intermediate between direct and fair simulation.

We study the border of GFQ preorders.
In our first attempt we generalize delayed simulation to \emph{delayed containment}.
While in simulation the two players take turns in selecting transitions,
in containment the game ends in one round:
First Spoiler picks an infinite path,
and then Duplicator has to match it with another infinite path.
The winning condition is delayed-like:
Every accepting state of Spoiler has to be matched by an accepting state of Duplicator,
possibly occurring later.
Therefore, in delayed containment Duplicator is much stronger than in simulation;
in other words, containment is coarser than simulation.
In fact, it is \emph{too coarse}:
We give a counterexample where delayed containment is not GFQ.
We henceforth turn our attention to finer preorders.

In our second attempt, we remedy to the deficiency above by introducing \emph{fixed-word delayed simulation},
an intermediate notion between simulation and containment.
In fixed-word simulation, Spoiler does not reveal the whole path in advance like in containment;
instead, she only declares the input word beforehand.
Then, the simulation game starts,
but now transitions can be taken only if they match the word fixed earlier by Spoiler.
Unlike containment, fixed-word delayed simulation is GFQ, as we show.

We proceed by looking at even coarser GFQ preorders.
We enrich fixed-word simulation by allowing Duplicator to use \emph{multiple pebbles},
in the style of \cite{etessami:hierarchy02}.
The question arises as whether Duplicator gains more power by ``hedging her bets''
when she already knows the input word in advance.
By using an ordinal ranking argument (reminiscent of \cite{KV01:weak}),
we establish that this is not the case,
and that the multipebble hierarchy collapses to the 1-pebble case,
i.e., to fixed-word delayed simulation itself.
Incidentally, this also shows that the whole delayed multipebble hierarchy from \cite{etessami:hierarchy02} is entirely contained in fixed-word delayed simulation---the containment being strict.

For what concerns the complexity of computing fixed-word simulation,
we establish that it is PSPACE-complete,
by a mutual reduction from B\"uchi automata universality.

With the aim of getting tractable preorders,
we then look at a different way of obtaining GFQ relations,
by introducing a theory of refinement transformers:
A \emph{refinement transformer} maps a preorder $\preceq$ to a coarser preorder $\preceq'$,
s.t., once $\preceq$ is known,
$\preceq'$ can be computed with only a polynomial time overhead.
The idea is to play a simulation-like game,
where we allow Duplicator to ``jump'' to $\preceq$-bigger states, called \emph{proxies},
after Spoiler has selected her transition.
Duplicator can then reply with a transition from the proxy instead of the original state.
We say that proxy states are \emph{dynamic}
in the sense that they depend on the transition selected by Spoiler.%
\footnote{Proxies are strongly related to mediators \cite{mediating:fsttcs2009}. We compare them in depth in Section~\ref{sec:proxy}.}
Under certain conditions,
we show that refinement transformers induce GFQ preorders.

Finally, we introduce \emph{proxy simulations},
which are novel polynomial time GFQ preorders obtained by applying refinement transformers to a concrete preorder $\preceq$,
namely, to \emph{backward} direct simulation (called reverse simulation in \cite{efficientltl:2000}).
We define two versions of proxy simulation, direct and delayed,
the latter being coarser than the former,
and both coarser than direct backward simulation.
Moreover, we show that the delayed variant can achieve quotients smaller than direct proxy simulation by an arbitrarily large factor.
Full proofs can be found in the appendix.

\paragraph{Related work.}

Delayed simulation \cite{etessami:etal:fairsimulations:05} has been extended to generalized automata \cite{piterman:generalized06},
to multiple pebbles \cite{etessami:hierarchy02},
to alternating automata \cite{wilke:fritz:simulations:05} and to the combination of the last two \cite{multipebble:concur10}.
Fair simulation has been used for state space reduction in \cite{fairminimization:02}.
The abstract idea of mixing forward and backward modes in quotienting can be traced back at least to \cite{raimi:phdthesis};
in the context of alternating automata,
it has been studied in \cite{mediating:fsttcs2009}.

\section{Preliminaries}
\label{sec:prelim}


\paragraph{Games.}

For a finite sequence $\pi = e_0e_1\cdots e_{k-1}$,
let $\card \pi = k$ be its length,
and let $\last \pi = e_{k-1}$ be its last element.
If $\pi$ is infinite, then take $\card \pi = \omega$.

A \emph{game} is a tuple $G = (P, P_0, P_1, p_I, \Gamma, \Gamma_0, \Gamma_1, W)$,
where $P$ is the set of positions, partitioned into disjoint sets $P_0$ and $P_1$,
$p_I \in P_0$ is the initial position, $\Gamma = \Gamma_0 \cup \Gamma_1$ is the set of moves,
where $\Gamma_0 \subseteq P_0 \times P_1$ and $\Gamma_1 \subseteq P_1 \times P_0$
are the set of moves of Player 0 and Player 1, respectively,
and $W \subseteq P_0^\omega$ is the winning condition.
A \emph{path} is a finite or infinite sequence of states
$\pi = p_0^0 p_0^1 p_1^0 p_1^1 \cdots$
starting in $p_I$,
such that, for all $i < \card \pi$,
$(p_i^0, p_i^1) \in \Gamma_0$ and $(p_i^1, p_{i+1}^0) \in \Gamma_1$.
\emph{Partial plays} and \emph{plays} are finite and infinite paths, respectively.
We assume that there are no dead ends in the game. 
A play is \emph{winning} for Player 1 iff $p_0^0p_1^0p_2^0 \cdots \in W$;
otherwise, is it winning for Player 0.

A \emph{strategy} for Player 0 is a partial function $\sigma_0 : (P_0P_1)^*P_0 \mapsto P_1$
s.t., for any partial play $\pi \in (P_0P_1)^*P_0$,
if $\sigma_0$ is defined on $\pi$,
then $\pi \cdot \sigma_0(\pi)$ is again a partial play.
A play $\pi$ 
is \emph{$\sigma_0$-conform}
iff, for every $i \geq 0$,
$p_i^1 = \sigma_0(p_0^0 p_0^1 \cdots p_i^0)$.
Similarly, a strategy for Player 1 is a partial function $\sigma_1 : (P_0P_1)^+ \mapsto P_0$
s.t., for any partial play $\pi \in (P_0P_1)^+$,
if $\sigma_1$ is defined on $\pi$,
then $\pi \cdot \sigma_1(\pi)$ is again a partial play.
A play $\pi$ 
is $\sigma_1$-conform
iff, for every $i \geq 0$,
$p_{i+1}^0 = \sigma_0(p_0^0 p_0^1 \cdots p_i^0 p_i^1)$.
While we do not require strategies to be total functions,
we do require that a strategy $\sigma$ is defined on all $\sigma$-conform partial plays.

A strategy $\sigma_i$ is a \emph{winning strategy} for Player $i$
iff all $\sigma_i$-conform plays are winning for Player $i$.
We say that Player $i$ wins the game $G$ if she has a winning strategy.

\paragraph{Automata.}

A \emph{nondeterministic B\"uchi automaton} (NBA) is a tuple $\NBA Q = \NBAtuple$,
where $Q$ is a finite set of states,
$\Sigma$ is a finite alphabet,
$I \subseteq Q$ is the set of initial states,
$F \subseteq Q$ is the set of final states
and $\Delta \subseteq Q \times \Sigma \times Q$ is the transition relation.
We also write $q \goesto a q'$ instead of $(q, a, q') \in \Delta$,
and just $q \goesto {} q'$ when $\exists a \in \Sigma \cdot q \goesto a q'$.
For two sets of states $\set q, \set q' \subseteq Q$,
we write $\set q \goto a \set q'$ iff
$\forall q' \in \set q' \cdot \exists q \in \set q \cdot q \goesto a q'$.\footnote{
This kind of backward-compatible transition had already appeared in \cite{fwbw_simulations1995}.
}
For a state $q \in Q$,
let $[q \in F] = 1$ if $q$ is accepting,
and $0$ otherwise.
We assume that every state is reachable from some initial state,
and that the transition relation is total.

For a finite or infinite sequence of states $\rho = q_0q_1 \cdots$ and an index $i \leq \card \rho$,
let $\countf{\rho}{i}$ be the number of final states occuring in $\rho$ up to (and including) the $i$-th element.
Formally,
$	\countf{\rho}{i} = \sum_{0 \leq k < i} [q_k \in F]$, with  $\countf \rho 0 = 0$.
Let $\countflast \rho = \countf \rho {\card \rho}$.
If $\rho$ is infinite,
then $\countflast{\rho} = \omega$ iff $\rho$ contains infinitely many accepting states.

Fix a finite or infinite word $w = a_0a_1\cdots$.
A \emph{path $\pi$ over $w$} is a sequence $q_0 \goesto {a_0} q_1 \goesto {a_1} q_2 \cdots$ of length $\card w + 1$.
A path is \emph{initial} if it starts in an initial state $q_0 \in I$,
it is a \emph{run} if it is initial and infinite,
and it is \emph{fair} if $\countflast \pi = \omega$.
An \emph{accepting run} is a run which is fair.
The \emph{language} $\omegalang {\NBA Q}$ of a NBA $\NBA Q$ is the set of infinite words which admit an accepting run,
i.e., $\omegalang {\NBA Q} = \{ w \in \Sigma^\omega \st \textrm{ there exists an accepting run $\pi$ over $w$}\}$.

\paragraph{Quotients.}

Let $\NBA{Q} = \NBAtuple$ be a NBA
and let $R$ be any binary relation on $Q$.
We say that $\eqsimrel_R$ is the \emph{equivalence induced by $R$} if
$\eqsimrel_R$ is the largest equivalence contained in the transitive and reflexive closure of $R$.
I.e., $\eqsimrel_R = R^* \!\!\cap (R^*)^{-1}$.
Let the function $[\cdot]_R: Q \mapsto 2^Q$ map each element $q\in Q$ to the equivalence class $[q]_R \subseteq Q$ it belongs to, i.e.,
$[q]_R := \{ q' \in Q \st q \eqsimrel_R q' \}$.
We overload $[P]_R$ on sets $P\subseteq Q$ by taking the set of equivalence classes.
When clear from the context, we avoid noting the dependence of $\eqsimrel$ and $[\cdot]$ on $R$.

An equivalence $\eqsimrel$ on $\NBA Q$ induces the \emph{quotient automaton}
$\quot{\NBA{Q}}{\eqsimrel}\!\! =\!\! \NBAquottuple{\eqsimrel}$,
where, for any $q, q' \in Q$ and $a\in\Sigma$,
$([q], a, [q']) \in \Delta_{\eqsimrel}$ iff $(q, a, q') \in \Delta$.
This is called a \emph{na\"ive} quotient
since both initial/final states and transitions are induced representative-wise.
When we quotient w.r.t. a relation $R$ which is not itself an equivalence,
we actually mean quotenting w.r.t. the induced equivalence $\approx$.
We say that $R$ is \emph{good for quotienting} (GFQ)
if quotienting $\NBA Q$ w.r.t. $R$ preserves the language,
that is, $\omegalang{\NBA Q} = \omegalang{\quot{\NBA{Q}}{\eqsimrel}}$.

\ignore{
We state a basic property of na\"ive quotients:
For two equivalences $\eqsimrel_0 \subseteq \eqsimrel_1$,
$\quot{\NBA{Q}}{\eqsimrel_1}$ recognizes at least as much as $\quot{\NBA{Q}}{\eqsimrel_0}$.
This is immediate, since, for any $q \in Q$, $[q]_{\eqsimrel_0} \subseteq [q]_{\eqsimrel_1}$.
Therefore, an accepting run $[q_0]_{\eqsimrel_0} \goesto {a_0} [q_1]_{\eqsimrel_0} \goesto {a_1} \cdots$ in $\quot{\NBA{Q}}{\eqsimrel_0}$
directly induces an accepting run $[q_0]_{\eqsimrel_1} \goesto {a_0} [q_1]_{\eqsimrel_1} \goesto {a_1} \cdots$ in $\quot{\NBA{Q}}{\eqsimrel_1}$.
}

\begin{lemma}\label{lem:quot_contains}
	For two equivalences $\eqsimrel_0, \eqsimrel_1$,
	if $\eqsimrel_0 \subseteq \eqsimrel_1$,
	then $\omegalang{\quot{\NBA{Q}}{\eqsimrel_0}} \subseteq \omegalang{\quot{\NBA{Q}}{\eqsimrel_1}}$.
	In particular, by letting $\eqsimrel_0$ be the identity,
	$\omegalang{\NBA Q} \subseteq \omegalang{\quot{\NBA{Q}}{\eqsimrel_1}}$.
\end{lemma}

\section{Quotienting with forward simulations}
\label{sec:fw}

In this section we study several generalizations of delayed simulation,
in order to investigate the border of good for quotienting (GFQ) forward-like preorders.
%
%
%
In our first attempt we introduce \emph{delayed containment},
which is obtained as a modification of the usual simulation interaction between players:
In the delayed containment game between $q$ and $s$ there are only two rounds.
Spoiler moves first and selects both an infinite word $w=a_0a_1\cdots$ and an infinite path $q_0 \goesto {a_0} q_1 \goesto {a_1} \cdots$ over $w$ starting in $q = q_0$;
then, Duplicator replies with an infinite path $s_0 \goesto {a_0} s_1 \goesto {a_1} \cdots$ over $w$ starting in $s = s_0$.
The winning condition is delayed-like:
$\forall i \cdot q_i \!\in\! F \!\!\implies\!\! \exists j \geq i \cdot s_j \!\in\! F$.
If Duplicator wins the delayed containment game between $q$ and $s$,
we write $q \decont s$.
Clearly, $\decont$ is a preorder implying language containment.
One might wonder whether delayed-containment is GFQ.
Unfortunately, this is not the case
(see Figure~\ref{fig:de-cont-not-GFQ} in the Appendix).
Therefore, $\decont$ is too coarse for quotienting,
and we shall look at finer relations.

\begin{lemma}\label{lem:decont_not_GFQ}
	$\decont$ is not a GFQ preorder.
\end{lemma}

\subsection{Fixed-word delayed simulation}

Our second attempt at generalizing delayed simulation still retains the flavour of containment.
While in containment $\decont$ Spoiler reveals both the input word $w$ and a path over $w$,
in \emph{fixed-word} simulation $\fxdesim$ Spoiler reveals $w$ only.
Then, after $w$ has been fixed,
the game proceeds like in delayed simulation,
with the proviso that transitions match symbols in $w$.%
\footnote{
The related notion of fixed-word \emph{fair} simulation clearly coincides with $\omega$-language inclusion.
}
Formally,
let $w=a_0 a_1 \cdots \in \Sigma^\omega$.
In the $w$-simulation game $\fxsimgame w q s$
the set of positions of Spoiler is $P_0 = Q \times Q \times \Nat$,
the set of positions of Duplicator is $P_1 = Q \times Q \times Q \times \Nat$
and $\triple q s 0$ is the initial position.
Transitions are determined as follows:
Spoiler can select a move of the form $(\triple q s i, \quadruple q s {q'} i) \in \Gamma^{w\textrm{-de}}_0$
if $q \goesto {a_i} q'$,
and Duplicator can select a move of the form $(\quadruple q s {q'} i, \triple {q'} {s'} {i+1}) \in \Gamma^{w\textrm{-de}}_1$
if $s \goesto {a_i} s'$.
Notice that the input symbol $a_i$ is fixed,
and it has to match the corresponding symbol in $w$.
The winning condition is 
$	W = \{ \triple {q_0} {s_0} 0 \triple {q_1} {s_1} 1 \cdots \st \forall i \cdot q_i \!\in\! F \!\!\implies\!\! \exists j \geq i \cdot s_j \!\in\! F \}$.
%
Let $q \wdesim w s$ iff Duplicator wins the $w$-simulation game $\fxsimgame w q s$,
and $q \fxdesim s$ iff $q \wdesim w s$ for all $w \in \Sigma^\omega$.
Clearly, fixed-word simulation is a preorder implying containment.
\begin{fact}
	$\fxdesim$ is a reflexive and transitive relation,
	and $\forall q,s \in Q \cdot q \fxdesim s\! \implies\! q \decont s$.
\end{fact}
\ignore{
\footnote{
Alternatively, one could define an explicit game for $\fxdesim$,
where Spoiler selects $w$ first,
and then the game evolves like in $\fxsimgame w q s$.
}
}

Unlike delayed containment,
fixed-word delayed simulation is GFQ.
Moreover, fixed-word delayed simulation quotients can be more succint than (multipebble) delayed simulation quotients
by an arbitrarily large factor.
See Figure~\ref{fig:fx-de-quot-succinct} in the Appendix.
\begin{theorem}\label{thm:fx_de_sim-GFQ}
	$\fxdesim$ is good for quotienting.
\end{theorem}

\paragraph{Complexity of delayed fixed word simulation.}

Let $q, s$ be two states in $\NBA Q$.
We reduce the problem of checking $q \fxdesim s$
to the universality problem of a suitable alternating B\"uchi product automaton (ABA) $\ABA A$.
We design $\ABA A$ to accept exactly those words $w$
s.t. Duplicator wins $\fxsimgame w q s$.
Then, by the definition of $\fxdesim$,
it is enough to check whether $\ABA A$ has universal language.
See \cite{vardi:alternating} (or Appendix~\ref{app:ABA}) for background on ABAs.


The idea is to enrich configurations in the fixed-word simulation game
by adding an obligation bit recording whether Duplicator has any pending constraint to visit an accepting state.
Initially the bit is 0,
and it is set to 1 whenever Spoiler is accepting;
a reset to 0 can occur afterwards,
if and when Duplicator visits an accepting state.

Let $\NBA Q = \NBAtuple$ be a NBA.
We define a product ABA
$\ABA A = (A, \Sigma, \delta, \alpha)$ as follows:
The set of states is $A = Q \times Q \times \{0,1\}$,
final states are of the form $\alpha = Q \times Q \times \{0\}$
and, for any $\langle q, s, b \rangle \in A$ and $a \in \Sigma$,
\vspace{-3.6ex}
%
\[
\vspace{-1.6ex}
\delta(\langle q, s, b \rangle, a) = \bigwedge_{q \goesto a q'} \bigvee_{s \goesto a s'} \langle q', s', b' \rangle, \quad \textrm{ where } 
b' = 	\left\{\begin{array}{ll}
		0 & \textrm{ if } s \in F \\
		1 & \textrm{ if } q \in F \wedge s \not\in F \\
		b & \textrm{ otherwise }
		\end{array}\right.
\]
%
It follows directly from the definitions that $q \fxdesim s$ iff $\omegalang{\langle q, s, 0 \rangle} = \Sigma^\omega$.
\ignore{
(Duplicator corresponds to Automaton, while Spoiler corresponds to Pathfinder.)%
\footnote{
Computing $\desim$ can be reduced to universality of \emph{unary} ABA's in exactly the same way (which is the same as emptiness in this case).
Given an alphabet $\Sigma$ and a fixed symbol $\hat a \in \Sigma$,
we define an ABA $\ABA A$ as above,
but with alphabet $\hat\Sigma = \{ \hat a \}$ and transition relation
\begin{align*}
\delta(\langle q, s, b \rangle, \hat a) = \bigwedge_{a \in \Sigma} \bigwedge_{q \goesto a q'} \bigvee_{s \goesto a s'} \langle q', s', b' \rangle . 
\end{align*}
Then, $q \desim s$ iff $\langle q, s, 0 \rangle$ is universal (over $\hat\Sigma$).
}
}
A reduction in the other direction is immediate already for NBAs:
In fact, an NBA $\NBA Q$ is universal iff $\NBA U \fxdesim \NBA Q$,
where $\NBA U$ is the trivial, universal one-state automaton with an accepting $\Sigma$-loop.
It is well-known that universality is PSPACE-complete for ABAs/NBAs \cite{kupfermanvardi:fair_verification}.

\begin{theorem}
	Computing fixed-word delayed simulation is PSPACE-complete.
\end{theorem}
	
\subsection{Multipebble fixed-word delayed simulation}

Having established that fixed-word simulation is GFQ,
the next question is whether we can find other natural GFQ preorders between fixed-word and delayed containment.
A natural attempt is to add a multipebble facility on top of $\fxdesim$.
Intuitively, when Duplicator uses multiple pebbles she can ``hedge her bets''
by moving pebbles to several successors.
This allows Duplicator to delay committing to any particular choice by arbitrarily many steps:
In particular, she can always gain knowledge on any \emph{finite} number of moves by Spoiler.
Perhaps surprisingly, we show that \emph{Duplicator does not gain more power by using pebbles}.
This is stated in Theorem~\ref{thm:fx_sim_equiv},
and it is the major technical result of this section.
It follows that, once Duplicator knows the input word in advance,
there is no difference between knowing only the next step by Spoiler,
or the next $l$ steps, for any finite $l > 1$.
Yet, if we allow $l = \omega$ lookahead, then we recover delayed containment $\decont$,
which is not GFQ by Lemma~\ref{lem:decont_not_GFQ}.
Therefore, w.r.t. to the degree of lookahead,
$\fxdesim$ is the coarsest GFQ relation included in $\decont$.


%
We now define the multipebble fixed-word delayed simulation.
Let $k \geq 1$ and $w = a_0 a_1 \cdots \in \Sigma^\omega$.
In the $k$-multipebble $w$-delayed simulation game $\kfxsimgame k w q s$
the set of positions of Spoiler is $Q \times 2^Q \times \Nat$,
the set of positions of Duplicator is $Q \times 2^Q \times Q \times \Nat$,
the initial position is $\triple q {\{s\}} 0$,
and transitions are: $(\triple q {\set s} i, \quadruple q {\set s} {q'} i) \in \Gamma_0$ iff 
$q \goesto {a_i} q'$,
and $(\quadruple q {\set s} {q'} i, \triple {q'} {\set s'} {i+1}) \in \Gamma_1$ iff 
$\set s \goto {a_i} \set s'$ and $\card {\set s'} \leq k$.

Before defining the winning set we need some preparation.
Given an infinite sequence $\pi = \triple {q_0} {\set s_0} 0 \triple {q_1} {\set s_1} 1 \cdots$
over $w=a_0 a_1\cdots$ and a round $j \geq 0$,
we say that a state $s \in \set s_j$ \emph{has been accepting} since some previous round $i \leq j$,
written $\acceptingsince{s}{\pi}{i}{j}$, iff either $s \in F$,
or $i < j$ and there exists $\hat s \in \set s_{j-1}$ s.t. $\hat s \goesto {a_{j-1}} s$
and $\acceptingsince{\hat s}{\pi}{i}{j-1}$.
We say that $\set s_j$ is \emph{good since round $i \leq j$},
written $\goodsince{\set s_j}{\pi}{i}{j}$,
iff at round $j$ every state $s \in \set s_j$ has been accepting since round $i$,
and $j$ is the least round for which this holds \cite{etessami:hierarchy02}.
Duplicator wins a play if, whenever $q_i \in F$
there exists $j \geq i$ s.t. $\goodsince{\set s_j}{\pi}{i}{j}$.
%
We write $q \wkdesim w k s$ iff Duplicator wins $\kfxsimgame k w q s$,
and we write $q \fxkdesim k s$ iff $\forall w \in \Sigma^\omega \cdot q \wkdesim w k s$.

Clearly, pebble simulations induce a non-decreasing hierarcy:
$\fxkdesim 1\ \subseteq\ \fxkdesim 2\ \subseteq \cdots$.
We establish that the hierarchy actually collapses to the $k=1$ level.
This result is non-trivial, since the delayed winning condition requires reasoning not only about the \emph{possibility} of Duplicator to visit accepting states in the future,
but also about exactly \emph{when} such a visit occurs.
Technically, our argument uses a ranking argument similar to \cite{KV01:weak} (see Appendix~\ref{app:fx_sim_equiv}),
with the notable difference that our ranks are \emph{ordinals} ($\leq\omega^2$), instead of natural numbers.
We need ordinals to represent how long a player can delay visiting accepting states,
and how this events nest with each other.
Finally, notice that the result above implies that the multipebble delayed simulation hierarchy of \cite{etessami:hierarchy02} is entirely contained in $\fxdesim$, and the containment is strict (Fig.~\ref{fig:fx-de-quot-succinct} in the appendix).
\begin{theorem}\label{thm:fx_sim_equiv}
	For any NBA $\NBA Q$, $k \geq 1$ and states $q,s \in Q$,
	$q \fxkdesim k s$ iff $q \fxdesim s$.
\end{theorem}

\section{Jumping-safe relations}
\label{sec:correctness}

In this section we present the general technique which is used throughout the paper to establish that preorders are GFQ.
We introduce \emph{jumping-safe relations}, which are shown to be GFQ (Theorem~\ref{thm:Upsilon_is_GFQ}).
In Section~\ref{sec:transformers} we use jumping-safety as an invariant when applying refinement transformers.
We start off with an analysis of acceping runs.

\paragraph{Coherent sequences of paths.}
Fix an infinite word $w \in \Sigma^\omega$.
Let $\Pi := \pi_0, \pi_1, \dots$ be an infinite sequence of longer and longer \emph{finite} initial paths in $\NBA Q$ over (prefixes of) $w$.
We are interested in finding a sufficient condition for the existence of an accepting run over $w$.
A necessary condition is that the number of final states in $\pi_i$ grows unboundedly as $i$ goes to $\omega$.
In the case of deterministic automata this condition is also sufficient:
Indeed, in a deterministic automaton there exists a unique run over $w$,
which is accepting exactly when the number of accepting stated visited by its prefixes goes to infinity.
In this case, we say that the $\pi_i$'s are \emph{strongly coherent} since they next path extends the previous one.
\begin{wrapfigure}{r}{0.3\textwidth}
	\centering
	\VCDraw{
		\begin{VCPicture}{(0,-0.1)(3,1)}


		\State[q]{(0,0)}{Q} \Initial[w]{Q}
		\FinalState[s]{(3,0)}{S}

		\LoopN[0.6]{Q}{a,b}
		\EdgeL{Q}{S}{a}
		\LoopN[0.6]{S}{a}

		\end{VCPicture}
	}
	\caption{Automaton $\NBA Q$.}
	\label{fig:unbounded_final_not_sufficient}
	\vspace{-2.1ex}
\end{wrapfigure}
Unfortunately, in the general case of nondeterministic automata
it is quite possible to have paths that visit arbitrarily many final states
but no accepting run exists.
This occurs because final states can appear arbitrarily late.
Indeed, consider Figure~\ref{fig:unbounded_final_not_sufficient}.
Take $w = a b a^2 b a^3 b \cdots$:
For every prefx $w_i = a b a^2 b \cdots a^i$
there exists a path $\pi_i = q q \cdots q \cdot s^i$ over $w_i$ visiting a final state $i$ times.
Still, $w \not\in \omegalang {\NBA Q}$.

Therefore, we forbid accepting states to ``clump away'' in the tail of the path.
We ensure this by imposing the existence of an infinite sequence of indices $j_0, j_1, \cdots$
s.t., for all $i$, and for all $k_i$ big enough,
the number of final states in $\pi_{k_i}$ up to the $j_i$-th state is at least $i$.
In this way, we are guaranteed that at least $i$ final states are present within $j_i$ steps in all but finitely many paths.

\begin{definition}\label{def:coherent_sequences}
Let $\Pi := \pi_0, \pi_1, \dots$ be an infinite sequence of finite paths.
We say that $\Pi$ is a \emph{coherent sequence of paths} if the following property holds:
\begin{align}\label{eq:inf_seq_prop}
	\forall i \cdot \exists j \cdot \exists h \cdot \forall k \geq h \cdot j < |\pi_k| \wedge \countf{\pi_k}{j} \geq i \ .
\end{align}
\end{definition}
%
%
%
\begin{lemma}\label{lem:increasing_seq_invariant_under_subseq}
	If $\Pi$ is coherent, then any infinite subsequence $\Pi'$ thereof is coherent.
\end{lemma}

We sketch below the proof that coherent sequences induce fair paths.
Let $\Pi = \pi_0, \pi_1, \dots$ be a coherent sequence of paths in $\NBA Q$.
Let $i = 1$, and let $j_1$ be the index witnessing $\Pi$ is coherent.
Since the $\pi_k$'s are branches in a finitely branching tree,
there are only a finite number of different prefixes of length $j_1$.
Therefore, there exists a prefix $\rho_1$ which is common to infinitely many paths.
Let $\Pi' = \pi'_0, \pi'_1, \dots$ be the infinite subsequence of $\Pi$
containing only suffixes of $\rho_1$.
Clearly $\rho_1$ contains at least $1$ final state,
and each $\pi'$ in $\Pi'$ extends $\rho_1$.
By Lemma~\ref{lem:increasing_seq_invariant_under_subseq},
$\Pi'$ is coherent.
For $i = 2$,
we can apply the reasoning again to $\Pi'$,
and we obtain a longer prefix $\rho_2$ extending $\rho_1$,
and containing at least $2$ final states.
Let $\Pi''$ be the coherent subsequence of $\Pi'$ containing only suffixes of $\rho_2$.
In this fashion,
we obtain an infinite sequence of \emph{strongly} coherent (finite) paths $\rho_1, \rho_2, \cdots$
s.t. $\rho_i$ extends $\rho_{i-1}$ and contains at least $i$ final states.
The infinite path to which the sequence converges is the fair path we are after.
\begin{lemma} \label{lem:increasing_seq_of_paths}
	Let $w \in \Sigma^\omega$ and $\pi_0, \pi_1, \dots$ as above.
	If $\pi_0, \pi_1, \dots$ is coherent,
	then there exists a fair path $\rho$ over $w$.
	Moreover, if all $\pi_i$'s are initial,
	then $\rho$ is initial.
\end{lemma}

\paragraph{Jumping-safe relations.}

We established that coherent sequences induce accepting paths.
Next, we introduce \emph{jumping-safe} relations,
which are designed to induce coherent sequences (and thus accepting paths) when used in quotienting.
The idea is to view a path in the quotient automaton as a jumping path in the original automaton,
where a ``jumping path'' is one that can take arbitrary jumps to equivalent states.
Jumping-safe relations allows us to transform the sequence of prefixes of an accepting jumping path
into a coherent sequence of non-jumping paths;
by Lemma~\ref{lem:increasing_seq_of_paths}, this induces a (nonjumping) accepting path.

Fix a word $w = a_0a_1\cdots \in \Sigma^\omega$,
and let $R$ be a binary relation over $Q$.
An \emph{$R$-jumping path} is an infinite sequence
\begin{align}\label{eq:jumping_path}
	\pi = q_0\ R\ q^F_0\ R\ \hat q_0 \goesto {a_0} q_1\ R\ q^F_1\ R\ \hat q_1 \goesto {a_1} q_2 \cdots,
\end{align}
and we say that $\pi$ is \emph{initial} if $q_0 \in I$,
and \emph{fair} if $q^F_i \in F$ for infinitely many $i$'s.

\begin{definition}
A binary relation $R$ is \emph{jumping-safe} iff
for any initial $R$-jumping path $\pi$
there exists an infinite sequence of initial finite paths $\pi_0, \pi_1, \dots$ over suitable prefixes of $w$
s.t. $\last{\pi_i}\ R\ q_i$ and,
if $\pi$ is fair, then $\pi_0, \pi_1, \dots$ is coherent.
\end{definition}
\begin{theorem}\label{thm:Upsilon_is_GFQ}
	Jumping-safe preorders are good for quotienting.
\end{theorem}

In Section~\ref{sec:transformers} we introduce refinement transformers,
which are designed to preserve jumping-safety.
Then, in Section~\ref{sec:proxy} we specialize the approach to \emph{backward direct simulation} $\bwdisim$ \cite{efficientltl:2000},
which provides an initial jumping-safe preorder, and which we introduce next:
$\bwdisim$ is the coarsest preorder s.t. $q \bwdisim s$ implies
1) $\forall (q' \goesto a q) \cdot \exists (s' \goesto  a s) \cdot q' \bwdisim s'$,
2) $q \in F \implies s \in F$, and
3) $q \in I \implies s \in I$.
\begin{fact}\label{fact:bwdisim_jumping_safe}
	$\bwdisim$ is jumping-safe and computable in polynomial time.
\end{fact}

\section{Refinement transformers}
\label{sec:transformers}

We study how to obtain GFQ preorders coarser than forward/backward simulation.
As a preliminary example, notice that it is not possible to generalize simultaneously both forward and backward simulations.
See the counterexample in Fig.~\ref{fig:fwbw_incorrect}, where
\begin{wrapfigure}{r}{0.4\textwidth}
	\centering
	\VCDraw{
		\begin{VCPicture}{(0,-0.5)(4,3)}
	
		\State[q_0]{(2,3)}{Q0} \Initial[nw]{Q0}
		\State[q_1]{(0,1.5)}{Q1}
		\State[q_2]{(2,1.5)}{Q2}
		\State[q_3]{(4,1.5)}{Q3}
		\FinalState[q_4]{(2,0)}{Q4}

		\EdgeR{Q0}{Q1}{a}
		\EdgeR{Q0}{Q2}{a}
		\EdgeL[0.4]{Q0}{Q3}{b}
		\EdgeR{Q1}{Q4}{a}
		\EdgeR{Q2}{Q4}{b}
		\EdgeL[0.2]{Q3}{Q4}{b}
		\LoopE[0.6]{Q4}{a}

		\end{VCPicture}
	}
	\caption{
	}
	\label{fig:fwbw_incorrect}
	
	\vspace{-5ex}
	
\end{wrapfigure}
any relation coarser than both forward and backward simulation is not GFQ.
Let $\bwdisimeq$ and $\fwdisimeq$ be backward and forward direct simulation equivalence, respectively.
We have $q_1 \bwdisimeq q_2 \fwdisimeq q_3$,
but ``glueing together'' $q_1, q_2, q_3$ would introduce the extraneous word $ba^\omega$. 
Therefore, one needs to choose whether to extend either forward or backward simulation.
The former approach has been pursued in the \emph{mediated preorders} of \cite{mediating:fsttcs2009}
(in the more general context of {alternating} automata).
Here, we extend backward refinements.
%
\label{sec:refinement_transformers}

We define a \emph{refinement transformer} $\tau_0$ mapping a relation $R$ to a new, coarser relation $\tau_0(R)$.
We present $\tau_0$ via a forward direct simulation-like game
where Duplicator is allowed to ``jump'' to $R$-bigger states---called \emph{proxies}. 
Formally, in the $\tau_0(R)$ simulation game
Spoiler's positions are in $Q\times Q$,
Duplicator's position are in $Q \times Q \times \Sigma \times Q$
and transitions are as follows:
Spoiler picks a transition $(\pair s q, \quadruple s q a {q'}) \in \Gamma_0$ simply when $q \goesto a {q'}$,
and Duplicator picks a transition $(\quadruple s q a {q'}, \pair {s'} {q'}) \in \Gamma_1$ iff
there exists a proxy $\hat s$ s.t. $s\ R\ \hat s$ and $\hat s \goesto a {s'}$.
The winning condition is:
%
$		\forall i \geq 0 \cdot q_i \in F \!\implies\! \hat s_i \in F$.
%
If Duplicator wins starting from the initial position $\pair s q$,
we write $s\ \tau_0(R)\ q$.
(Notice that we swapped the usual order between $q$ and $s$ here.)

\ignore{
\begin{remark}
	If we take $R$ to be the identity relation,
	we obtain the usual forward direct simulation $\fwdisim = \left[\tau_0(R)\right]^{-1}$.
\end{remark}
}

\begin{lemma}\label{lem:tau0_summary}
	For a preorder $R$, $R \subseteq R \circ \tau_0(R) \subseteq \tau_0(R)$.
\end{lemma}

Unfortunately, $\tau_0(R)$ is not necessarily a transitive relation.
Therefore, it is not immediately clear how to define a suitable equivalence for quotienting.
Figure~\ref{fig:fwbw_incorrect} shows that taking the transitive closure of $\tau_0(R)$ is incorrect---%
already when $R$ is direct backward simulation $\bwdisim$:
Let $\preceq = \tau_0(\bwdisim)$ and let $\approx = \preceq \cap \preceq^{-1}$. 
We have $q_3 \approx q_2 \approx q_1\ \preceq\ q_3$,
but $q_3 \not \preceq q_1$, and forcing $q_1 \approx q_3$ is incorrect, as noted earlier.

Thus, $\tau_0(R)$ is not GFQ
and we need to look at its transitive fragments.
Let $T \subseteq \tau_0(R)$.
We say that $R$ is \emph{$F$-respecting} if $q\ R\ s \wedge q \in F\!\! \implies\!\! s \in F$,
that $T$ is \emph{self-respecting} if Duplicator wins by never leaving $T$,
that $T$ is \emph{appealing} if transitive and self-respecting,
and that $T$ \emph{improves on $R$} if $R \subseteq T$.
\begin{theorem}\label{thm:tau0_GFQ}
	Let $R$ a $F$-respecting preorder,
	and let $T \subseteq \tau_0(R)$ be an appealing, improving fragment of $\tau_0(R)$.
	If $R$ is jumping-safe, then $T$ is jumping-safe.
\end{theorem}
In particular, by Theorem~\ref{thm:Upsilon_is_GFQ}, $T$ is GFQ.
Notice that requiring that $R$ is GFQ is not sufficient here,
and we need the stronger invariant given by jumping-safety.

Given an appealing fragment $T \subseteq \tau_0(R)$,
a natural question is whether $\tau_0(T)$ improves on $\tau_0(R)$,
so that $\tau_0$ can be applied repeatedly to get bigger and bigger preorders.
We see in the next lemma that this is not the case.
\begin{lemma}\label{lem:tau0_notimproving}
	For any reflexive $R$, let $T \subseteq \tau_0(R)$ be any appealing fragment of $\tau_0(R)$.
	Then, $\tau_0(T) \subseteq \tau_0(R)$.
\end{lemma}

\paragraph{Efficient appealing fragments.}

By Theorems \ref{thm:Upsilon_is_GFQ} and \ref{thm:tau0_GFQ},
appealing fragments of $\tau_0$ are GFQ.
Yet, we have not specified any method for obtaining these.
Ideally, one looks for fragments having maximal cardinality
(which yelds maximal reduction under quotienting),
but finding them is computationally expensive. 
Instead, we define a new transformer $\tau_1$ which is guaranteed to produce only appealing fragments,%
\footnote{
$\tau_1$ needs not be the only solution to this problem:
Other ways of obtaining appealing fragments of $\tau_0$ might exist.
For this reason, we have given a separate treatment of $\tau_0$ in its generality,
together with the general correctness statement (Theorem~\ref{thm:tau0_GFQ}).
}
which, while not maximal in general,
are maximal amongst all \emph{improving} fragments (Lemma~\ref{lem:tau1_maximal}).

The reason why $\tau_0(R)$ is not transitive
is that only Duplicator is allowed to make ``$R$-jumps''.
This asymmetry is an obstacle to compose simulation games.
We recover transitivity by allowing Spoiler to jump as well,
thus restoring the symmetry.
Formally, the $\tau_1(R)$ simulation game is identical to the one for $\tau_0(R)$,
the only difference being that also Spoiler is now allowed to ``jump'',
i.e., she can pick a transition $(\pair s q, \quadruple s q a {q'})\!\in\!\Gamma_0$
iff there exists $\hat q$ s.t. $q\ R\ \hat q$ and  $\hat q \goesto a {q'}$.
The winning condition is:
%
$	\forall i \geq 0 \cdot \hat q_i \!\in\! F \!\!\implies\!\! \hat s_i \!\in\! F$. 
%
Let $s\ \tau_1(R)\ q$ if Duplicator wins from position $\pair s q$.
%
It is immediate to see that $\tau_1(R)$ is an appealing fragment of $\tau_0(R)$,
and that $\tau_1$ is improving on transitive relations $R$'s.
Thus, for a preorder $R$,
$R \subseteq \tau_1(R) \subseteq \tau_0(R)$.
By Theorems \ref{thm:Upsilon_is_GFQ} and \ref{thm:tau0_GFQ}, $\tau_1(R)$ is GFQ (if $R$ is $F$-respecting).

It turns out that $\tau_1(R)$ is actually the \emph{maximal} appealing, improving fragment of $\tau_0(R)$.
This is non-obvious,
since the class of appealing $T$'s is not closed under union%
---still, it admits a maximal element.
Therefore, $\tau_1$ is an optimal solution to the problem of finding appealing, improving fragments of $\tau_0(R)$.

\begin{lemma}\label{lem:tau1_maximal}
	For any $R$,
	let $T \subseteq \tau_0(R)$ be any appealing fragment of $\tau_0(R)$.
	If $R \subseteq T$ (i.e., $R$ is improving),
	then $T \subseteq \tau_1(R)$.
\end{lemma}

\subsection{Delayed-like refinement transformers}
\label{sec:taude}

We show that the refinement transformer approach can yield relations even coarser than $\tau_1$.
Our first attempt is to generalize the direct-like winning condition of $\tau_0$ to a delayed one.
Let $\taudenaught$ be the same as $\tau_0$ except for the different winning condition,
which now is:
%
$	\forall i \geq 0 \cdot q_i \!\in\! F \!\!\implies\!\! \exists j\geq i \cdot \hat s_j \!\in\! F$.
%
Clearly, $\taudenaught$ inherits the same transitivity issues of $\tau_0$.
Unfortunately, the approach of taking appealing fragments is not sound here,
due to the weaker winning condition.
See Figure~\ref{fig:taudenaught_not_GFQ} in the Appendix for a counterexample.


We overcome these issues by dropping $\taudenaught$ altogether,
and directly generalize $\tau_1$ (instead of $\tau_0$) to a delayed-like notion.
The \emph{delayed refinement transformer} $\taude$ is like $\tau_1$,
except for the new winning condition:
%
$	\forall i \geq 0 \cdot \hat q_i \!\in\! F \!\!\implies\!\! \exists j\geq i \cdot \hat s_j \!\in\! F$. 
%
Notice that	$\taude(R)$ is at least as coarse as $\tau_1(R)$,
and incomparable with $\tau_0(R)$.
Once $R$ is given, $\taude(R)$ can be computed in polynomial time.
See Appendix~\ref{app:computing_taude}.

\begin{lemma}\label{lem:taude_transitive}
	For any $R$,
	$\taude(R)$ is transitive.
\end{lemma}

\begin{theorem}\label{thm:taude_GFQ}
	If $R$ is a jumping-safe $F$-respecting preorder,
	then $\taude(R)$ is jumping-safe.
\end{theorem}

\section{Proxy simulations}
\label{sec:proxy}

We apply the theory of transformers from Section~\ref{sec:transformers} to a specific $F$-respecting preorder,
namely backward direct simulation,
obtaining \emph{proxy simulations}.
Notice that proxy simulation-equivalent states need not have the same language;
yet, proxy simulations are GFQ (and computable in polynomial time).

\subsection{Direct proxy simulation}

Let \emph{direct proxy simulation}, written $\diproxysim$,
be defined as $\diproxysim := [\tau_1(\bwdisim)]^{-1}$.
%
\begin{theorem}
	$\diproxysim$ is a polynomial time GFQ preorder
 	at least as coarse as $(\bwdisim)^{-1}$.
\end{theorem}

\paragraph{Proxies vs mediators.}

Direct proxy simulation and mediated preorder \cite{mediating:fsttcs2009} are in general incomparable.
While proxy simulation is at least as coarse as backward direct simulation,
mediated preorder is at least as coarse as \emph{forward} direct simulation.
(We have seen in Section~\ref{sec:transformers} that this is somehow unavoidable,
since one cannot hope to generalize simultaneously both forward and backward simulation.)

One notable difference between the two notions is that proxies are ``dynamic'', while mediators are ``static'':
While Dupicator chooses the proxy only \emph{after} Spoiler has selected her move,
mediators are chosen uniformly w.r.t. Spoiler's move.

\begin{figure}
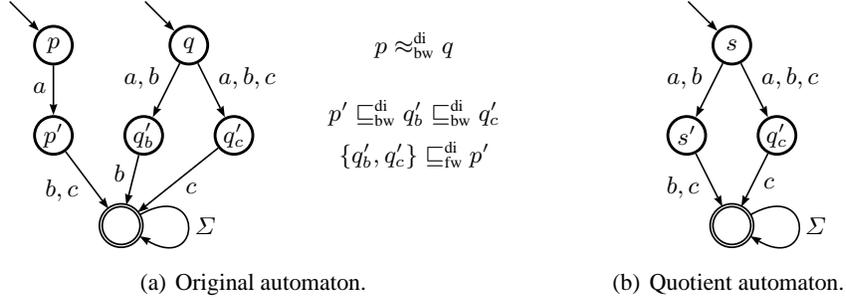

	
	\centering
	\subfigure[Original automaton.] {
	
		\label{fig:direct_proxy_simulation_a}
	
		\VCDraw{
			\begin{VCPicture}{(-1,-1)(10,4)}
	
			\State[p]{(0,4)}{P} \Initial[nw]{P}
			\State[p']{(0,2)}{P'}
			\FinalState[]{(1.5,0)}{P''}
		
			\EdgeR{P}{P'}{a}
			\EdgeR{P'}{P''}{b,c}
			\LoopE[0.5]{P''}{\Sigma}

			\State[q]{(3,4)}{Q} \Initial[nw]{Q}
			\State[q'_b]{(2,2)}{Q'b}
			\State[q'_c]{(4,2)}{Q'c}

			\EdgeR{Q}{Q'b}{a,b}
			\EdgeL{Q}{Q'c}{a,b,c}
			\EdgeR{Q'b}{P''}{b}
			\EdgeL{Q'c}{P''}{c}

			\ChgStateLineStyle{none}
			\State[p \bwdisimeq q ]{(8,4)}{lbl}
			\State[\begin{array}{cc}p' \bwdisim q'_b \bwdisim q'_c \\ \rule{0pt}{3ex}\{q'_b, q'_c\} \fwdisim p' \end{array}]{(8,2)}{lbl}
		
			\RstStateLineStyle
			\end{VCPicture}
		}
	}
	\qquad\qquad
	\subfigure[Quotient automaton.] {
	
		\label{fig:direct_proxy_simulation_b}
	
		\VCDraw{
			\begin{VCPicture}{(-1.5,-1)(3.5,4)}

			\State[s]{(1,4)}{Q} \Initial[nw]{Q}
			\State[s']{(0,2)}{Q'b}
			\State[q'_c]{(2,2)}{Q'c}
			\FinalState[]{(1,0)}{P''}
		
			\LoopE[0.5]{P''}{\Sigma}

			\EdgeR{Q}{Q'b}{a,b}
			\EdgeL{Q}{Q'c}{a,b,c}
			\EdgeR{Q'b}{P''}{b,c}
			\EdgeL{Q'c}{P''}{c}

			\end{VCPicture}
		}
	}
	
	\caption{
	Direct proxy simulation quotients.
	}
	\label{fig:direct_proxy_simulation}
	
	\vspace{-4ex}
	
\end{figure}

In Figure~\ref{fig:direct_proxy_simulation_a} we show a simple example where $\diproxysim$ achieves greater reduction.
Recall that mediated preorder $M$ is always a subset of $\fwdisim \! \circ \bwdisimrev$ \cite{mediating:fsttcs2009}.
In the example, static mediators are just the trivial ones already present in forward simulation.
Thus, $\fwdisim \! \circ \bwdisimrev = \fwdisim$
and mediated preorder $M$ collapses to forward simulation.
On the other side, $p \diproxysimeq q$ and $p' \diproxysimeq q'_b$.
Letting $s = [p,q]$ and $s' = [p',q'_b]$,
we obtain the quotient in Figure~\ref{fig:direct_proxy_simulation_b}.


\subsection{Delayed proxy simulation}

Another difference between the mediated preorder approach \cite{mediating:fsttcs2009} and the approach through proxies
is that proxies directly enable a delayed simulation-like generalization (see Section~\ref{sec:taude}).
Again, we fix backward delayed simulation $\bwdisim$ as a starting refinement,
and we define \emph{delayed proxy simulation} as $\deproxysim := [\taude(\bwdisim)]^{-1}$.

\begin{wrapfigure}{r}{0.45\textwidth}
	\centering
	\VCDraw{
		\begin{VCPicture}{(0.5,-1.5)(5.5,5)}

		\LargeState

		\FinalState[q_0]{(3,6)}{q0} \Initial[nw]{q0}
		\State[q_{k\text{-}1}]{(0,4)}{qk1}
		\State[s]{(3,2.5)}{s}
		\State[q_1]{(6,4)}{q1}
		\State[q_2]{(6,1)}{q2}
		\State[q_3]{(3,-1)}{q3}
		\ChgStateLineStyle{none}
		\State[\rotatebox{-15}{$\ddots$}]{(0,1)}{dots}

		\LoopN[0.82]{q0}{b}
		\ArcL{q0}{q1}{a}
		\ArcL{q1}{q2}{a}
		\ArcL{q2}{q3}{a}
		\ArcL{q3}{dots}{a}
		\ArcL{dots}{qk1}{a}
		\ArcL{qk1}{q0}{a,b}

		\ChgCLoopAngle{30}
		\CLoopSW[0.6]{s}{b}

		\ArcL{s}{q0}{b}
		\ArcL{q0}{s}{a}
		\ArcL{s}{q1}{b}
		\ArcL{q1}{s}{a}
		\ArcL{s}{q2}{b}
		\ArcL{q2}{s}{a}
		\ArcL{s}{q3}{b}
		\ArcL{q3}{s}{a}
		\ArcL{s}{qk1}{b}
		\ArcL{qk1}{s}{a}


		\end{VCPicture}
	}
	\caption{
	}
	\label{fig:dexysim_better}
	
	\vspace{-10ex}
	
\end{wrapfigure}
\begin{theorem}
	$\deproxysim$ is a polynomial time GFQ preorder. 
\end{theorem}
Notice that delayed proxy simulation is at least as coarse as direct proxy simulation.
Moreover, quotients w.r.t. $\deproxysim$ can be smaller than direct forward/backward/proxy and delayed simulation quotients by an arbitrary large factor. See Figure~\ref{fig:dexysim_better}:
Forward delayed simulation is just the identity,
and no two states are direct backward or proxy simulation equivalent.
But $q_i \bwdisim s$ for any $0 \!<\! i \!\leq\! k - 1$.
This causes any two outer states $q_i, q_j$ to be $\deproxysim$-equivalent.
Therefore, the $\deproxysim$-quotient automaton has only 2 states.

\section{Conclusions and Future Work}

We have proposed novel refinements for quotienting B\"uchi automata:
fixed-word delayed simulation and direct/delayed proxy simulation.
Each one has been shown to induce quotients smaller than previously known notions.

We outline a few directions for future work.
First, we would like to study practical algorithms for computing fixed-word delayed simulation,
and to devise efficient fragments thereof---one promising direction is to look at self-respecting fragments,
which usually have lower complexity.
Second, we would like to exploit the general correctness argument developed in Section~\ref{sec:correctness}
in order to get efficient purely backward refinements (coarser than backward direct simulation).
Finally, experiments on cases of practical interest are needed for an empirical evaluation of the proposed techniques.

{\bf\noindent Acknowledgment.}
We thank Richard Mayr and Patrick Totzke for helpful discussions,
and two anonymous reviewers for their valuable feedback.

\bibliography{citeulike}

\newpage
\appendix

\section{Proofs and additional material for Section~\ref{sec:fw}}
\label{app:fw}

\begin{figure}
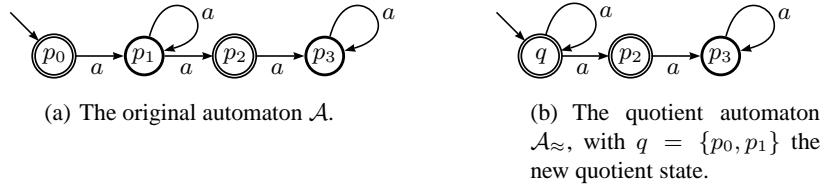

	
	\centering
	\subfigure[The original automaton $\NBA A$.] {
	
		\VCDraw{
			\begin{VCPicture}{(0,-1)(6,1)}
		
			\FinalState[p_0]{(0,0)}{P0} \Initial[nw]{P0}
			\State[p_1]{(2,0)}{P1}
			\FinalState[p_2]{(4,0)}{P2}
			\State[p_3]{(6,0)}{P3}
	
			\EdgeR{P0}{P1}{a}
			\LoopNE[0.6]{P1}{a}
			\EdgeR{P1}{P2}{a}
			\EdgeR{P2}{P3}{a}
			\LoopNE[0.6]{P3}{a}
	
			\end{VCPicture}
		}
	
		\label{fig:de-cont-not-GFQ-orig}
	
	}
	\qquad\qquad\qquad\qquad
	\subfigure[The quotient automaton $\quot{\NBA A}{\approx}$, with $q = \{p_0, p_1\}$ the new quotient state.] {
		\VCDraw{
			\begin{VCPicture}{(0,-1)(6,1)}
	
			\FinalState[q]{(0,0)}{P0} \Initial[nw]{P0}
			\FinalState[p_2]{(2,0)}{P2}
			\State[p_3]{(4,0)}{P3}

			\LoopNE[0.6]{P0}{a}
			\EdgeR{P0}{P2}{a}
			\EdgeR{P2}{P3}{a}
			\LoopNE[0.6]{P3}{a}

			\end{VCPicture}
		}
		\label{fig:de-cont-not-GFQ-quot}
	}
	\caption{
	An example showing that delayed containment cannot be employed for quotienting.
	We have that $p_0$ is delayed containment equivalent to $p_1$.
	Notice that the automaton $\NBA A$ in \subref{fig:de-cont-not-GFQ-orig} does not accept $a^\omega$,
	but the quotient automaton $\quot{\NBA A}{\approx}$ in \subref{fig:de-cont-not-GFQ-quot},
	obtained by identifying $p_0$ and $p_1$,
	does.
	\ignore{
	Incidentally,
	$p_0$ is not fixed-word delayed simulation equivalent to $p_1$ 
	(as it is easy to check).
	Thus, the example also shows that the inclusion between fixed-word simulation and containment is, in general, strict.
	\newline
	FIXME: moreover, the same example also shows that de-containment cannot be used in the Ramsey-based approach of the CAV paper...
	}
	}
	
	\label{fig:de-cont-not-GFQ}
	
\end{figure}

We postpone the proof of Theorem~\ref{thm:fx_de_sim-GFQ} until Section~\ref{app:fx_de_sim-GFQ}.

\begin{figure}
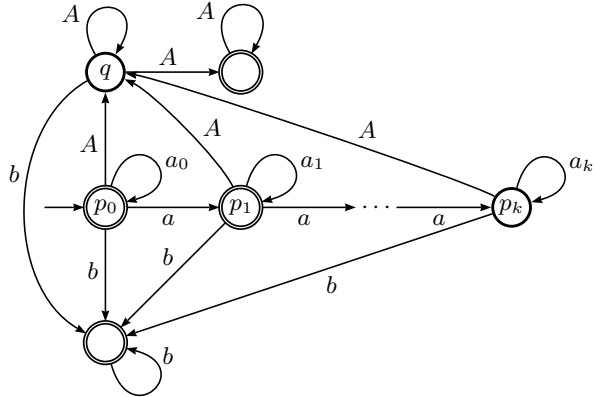

	
	\centering
	
		\VCDraw{
			\begin{VCPicture}{(0,-2)(9,6)}


			\FinalState{(0,-1)}{Q1}

			\FinalState[p_0]{(0,2)}{P0} \Initial[w]{P0}
			\FinalState[p_1]{(3,2)}{P1}
			\ChgStateLineStyle{none}
			\State[\cdots]{(6,2)}{dots}
			\RstStateLineStyle
			\State[p_k]{(9,2)}{Pk}

			\State[q]{(0,5)}{Q}
			\FinalState{(3,5)}{Q2}

			\LoopNE[0.6]{P0}{a_0}
			\EdgeR{P0}{P1}{a}
			\EdgeL{P0}{Q}{A}
			\EdgeR{P0}{Q1}{b}
			
			\LoopNE[0.6]{P1}{a_1}
			\EdgeR{P1}{dots}{a}
			\VArcR{arcangle=-20,ncurv=0.4}{P1}{Q}{A}
			\EdgeR{P1}{Q1}{b}
			
			\EdgeR{dots}{Pk}{a}
			
			\LoopNE[0.6]{Pk}{a_k}
			\VArcR{arcangle=-10,ncurv=0.2}{Pk}{Q}{A}
			\EdgeL{Pk}{Q1}{b}
			
			\LoopN{Q}{A}
			\EdgeL{Q}{Q2}{A}
			\VArcR{arcangle=-60, ncurv=0.75}{Q}{Q1}{b}
			
			\LoopN{Q2}{A}
			
			\LoopSE[0.78]{Q1}{b}

			\end{VCPicture}
		}
	
	\caption{
	Fixed-word delayed simulaton quotients can achieve arbitrarily high compression ratios.
	}
	
	\label{fig:fx-de-quot-succinct}
	
\end{figure}

\subsection{Alternating B\"uchi automata}
\label{app:ABA}

Below, we give a self-contained definition of alternating B\"uchi automata.
The syntax follows the presentation of \cite{vardi:alternating},
while tbe semantics adheres to \cite{wilke:fritz:simulations:05}.

For a set $A$, let $\mathcal B^+(A)$ be the set of positive boolean formulas over $A$,
that is, $\mathcal B^+(A)$ is the smallest set containing $A \cup \{\true, \false \}$ and closed under the operations $\wedge$ and $\vee$.
For a formula $\varphi \in \mathcal B^+(A)$ and a set $X \subseteq A$,
we write $X \models \varphi$ iff the truth assignment assigning $\true$ to elements in $X$ and $\false$ to the elements in $A \setminus X$
satisfies $\varphi$.
An alternating B\"uchi automaton (ABA) is a tuple $\ABA A = \ABAtuple$,
where $A$ is a finite set of states,
$\Sigma$ is a finite set of input symbols,
$\delta: A \times \Sigma \mapsto \mathcal B^+(A)$ is the transition relation
and $\alpha \subseteq A$ is the set of accepting states.
Acceptance of an ABA $\ABA A$ is best defined via games \cite{wilke:fritz:simulations:05}. 
In this context, the two players are usually named Automaton and Pathfinder.
Given an infinite word $w = a_0 a_1 \cdots \in \Sigma^\omega$ and a distinguished starting state $p_I$,
the \emph{acceptance game for $w$ from $p_I$} is a game where
$P_0 = Q \times \omega$ is the set of Automaton's positions,
$P_1 = Q \times 2^Q \times \omega$ is the set of Pathfinder's positions,
$(p_I, 0)$ is the initial position,
and transitions are determined as follows.
Automaton can select a transition $(\pair p i, \triple p {\set p'} i)$
iff $\set p' \models \delta(p, a_i)$,
and Pathfinder can select a transition $(\triple p {\set p'} i, \pair {p'} {i+1})$
iff $p' \in \set p'$.
Finally, the winning condition consists of those paths visiting $\alpha$ infinitely often.
A state $p \in A$ accepts $w \in \Sigma^\omega$ iff Automaton wins the acceptance game for $w$ from $p$.
A state $p$ is universal iff it accepts every word $w \in \Sigma^\omega$.

\subsection{Proof of Theorem~\ref{thm:fx_sim_equiv}}
\label{app:fx_sim_equiv}

\paragraph{Preliminaries on ordinals.}

Let $\omega$ be the least infinite ordinal,
and let $\omega_1$ be the set of all countable ordinals.
We denote abitrary ordinals by $\alpha$ or $\beta$,
and limit ordinals by $\lambda$ or $\mu$.
In this paper, $0$ is considered to be a limit ordinal.

\paragraph{Preliminaries on trees.}

Let $[n] = \{ 0, 1, \dots, n - 1 \}$.
A tree domain is a non-empty, prefix-closed subset $V$ of $[n]^*$.
With $\prefixle$ we denote the prefix order on words;
if $u \prefixle u'$, then $u'$ is called a descendant of $u$
and $u$ is an ancestor of $u'$.
In particular, if $u' = uc$ for some $c\in\Nat$,
then $u'$ is a child of $u$.
A (labelled) $L$-tree is a pair $(V, t)$, where $V$ is a tree domain and $t : V \mapsto L$
is a mapping which assigns a label from $L$ to any node in the tree.

\paragraph{The ranking construction.}

Len $\NBA{Q} = \NBAtuple$ be an automaton,
and let $n$ be the cardinality of $Q$.
Given an infinite word $w=a_0a_1\cdots\in\Sigma^\omega$,
we associate to any state $q \in Q$ a tree domain $T^w_q$ and a $Q$-tree $(T^w_q, t^w_q)$,
the \emph{unravelling} of $Q$ from $q$ while reading $w$,
by applying the following two rules:
\begin{itemize}
	\item $\varepsilon \in T^w_q$ and $t^w_q(\epsilon) = q$.
	\item If $u$ has length $i$, $u\in T^w_q$, $t^w_q(u) = p$ and $\Delta(p, a_i) = \{ p'_0, p'_1, \dots, p'_{k-1} \}$,
	then, for any $j$ s.t. $0 \leq j < k$, $uj \in T^w_q$ and $t^w_q(uj) = p'_j$.
\end{itemize}
%

It is easy to see that if two nodes at the same level have the same label,
then they generate isomorphic subtrees.
Therefore, we can ``compress'' $(T^w_q, t^w_q)$ into an infinite DAG $G^w_q=(V, E)$, where
$V \subseteq Q \times \Nat$ is such that $\pair q l \in V$ iff there exists a node in $(T^w_q, t^w_q)$ at level $l$ with label $q$,
and $(\pair q l, \pair {q'} {l+1}) \in E$ iff there exist two nodes $u$ and $u'$,
labelled with $q$ and $q'$, respectively,
s.t. $u'$ is a child of $u$ in $(T^w_q, t^w_q)$.
We say that a vertex $\pair q l$ is \emph{accepting} iff $q\in F$.

For any $G \subseteq G^w_q$,
we say that a vertex $\pair q l$ is a \emph{dead end} in $G$ iff it has no successor in $G$,
and we say that it is \emph{inert} in $G$ iff no accepting vertex can be reached from $\pair q l$ in $G$.
In particular, an inert vertex is not accepting.
The \emph{girth} of $G$ at level $l$ is the maximal number of vertices of the form $\pair q l$ in $G$,
and the \emph{width} of $G$ is the maximal girth over infinitely many levels.

We build a nonincreasing transfinite sequence of DAGs $\{ G_\alpha \st \alpha < \omega_1 \}$ as follows:
\begin{align*}
	G_0 &= G^w_q \\
	G_{\alpha+1} &= G_\alpha\ \setminus\ \{ \pair q l \st \pair q l \textrm{ is a dead end in } G_\alpha \} \\
	G_\lambda &= H_\lambda \ \setminus\ \{ \pair q l \st \pair q l \textrm{ is inert in } H_\lambda \},
\end{align*}
where, for any ordinal $\alpha$,
$H_\alpha = \bigcap_{\beta < \alpha} G_\beta$.
Notice that $H_{\alpha + 1} = G_\alpha$;
and $\alpha \leq \beta$ implies $G_\beta \subseteq G_\alpha$.

Assume that there is no path in $G^w_q$ with an infinite number of accepting vertices.
As a direct consequence of K\"onig's Lemma,
we have that when moving from $H_\lambda$ to $G_\lambda$
an infinite path is removed from the graph.
Therefore, the width of $G_\lambda$ is strictly less than the width of $H_\lambda$.
Since the width of $G^w_q$ is (uniformly) bounded by $\omega$,
it follows that $H_{\omega^2}$ is empty,
and thus $G_{\omega^2}$ is empty as well.
Therefore, each vertex is either a dead end in $G_\alpha$ or inert in $H_\lambda$.
In the former case $\pair q l$ is in $G_\alpha$ but not in $G_{\alpha+1}$,
whereas in the latter case $\pair q l$ is in $H_\lambda$ but not in $G_\lambda$.
Accordingly, we associate an ordinal \emph{rank} to every vertex $\pair q l$ in $G^w_q$:
%
\begin{align}\tag*{$(\mathsf{Rank})$}\label{eq:rank}
	\rank {q, l} = \sup_{\alpha < \omega^2} \{ \alpha \st \pair q l \in H_\alpha \} \ . 
\end{align}

Therefore, under the assumption that $G^w_q$ does not contain any fair path,
no vertex receives rank $\omega^2$.
On the other side, if $G^w_q$ contained a fair path,
then there exists an infinite path of non-inert vertices starting at $\pair q 0$:
In this case, the ranking construction ``does not terminate''
and stabilizes (at most) at a nonempty $G_{\omega^2} = G_\alpha \neq \emptyset$ for all $\alpha \geq \omega^2$.
Thus, vertices in $G_{\omega^2}$ would receive rank $\omega^2$ according to \ref{eq:rank}.
Since no conflict can arise,
we drop any assumption about fair paths thereafter,
and we uniformly apply \ref{eq:rank} in either case.

\begin{remark}
	It is clear from \ref{eq:rank} that no ordinal larger than $\omega^2$ is actually used in our construction.
	In fact, we could have given an equivalent presentation in terms of pairs of natural numbers ordered lexicographically.
	However, we have chosen to use ordinals $\leq\omega^2$ for technical convenience.
\end{remark}

\begin{remark}
	A vertex $\pair q l$ is in $H_\alpha$ iff it has rank $\geq \alpha$,
	and it is in not in $G_\alpha$ iff it has rank $\leq \alpha$.
	Therefore, $\rank {q, l} = \alpha \iff \pair q l \in H_\alpha \setminus G_\alpha$.
	
	
\end{remark}

\begin{lemma}\label{lem:rank_accepting}
	If a vertex $\pair q l$ is accepting, then it has rank $\alpha + 1$.
	Furthermore, if it has rank $\lambda + 1$, then it is accepting.
\end{lemma}
\begin{proof}
	The first part follows from the fact that an accepting vertex $\pair q l$ is not inert:
	Therefore, $\pair q l$ is a dead end in $G_\alpha$,
	so $\pair q l \not \in G_{\alpha + 1}$ and $\rank {q, l} = \alpha + 1$.
	
	For the second part,
	assume $\rank {q, l} = \lambda + 1$, i.e.,
	$\pair q l \in G_\lambda \setminus G_{\lambda + 1}$.
	Therefore, $\pair q l$ is a dead end in $G_\lambda$.
	Since $G_\lambda \subseteq H_\lambda$,
	$\pair q l$ is in $H_\lambda$ as well.
	But $H_\lambda$ has no dead ends,
	therefore $\pair q l$ has at least one successor $\pair {q'} {l+1}$ in $H_\lambda$.
	But $\pair q l$ is a dead end in $G_\lambda$,
	therefore any such successor $\pair {q'} {l+1}$ is not in $G_\lambda$.
	Therefore, $\pair {q'} {l+1}$ is inert in $H_\lambda$.
	
	By contradiction, assume $\pair q l$ that is not accepting.
	Since it has only inert successors $\pair {q'} {l+1}$ in $H_\lambda$,
	it is itself inert in $H_\lambda$.
	But $\pair q l \in G_\lambda$,
 	so $\pair q l$ is \emph{not} inert in $H_\lambda$.
	This is a contradiction, therefore $\pair q l $ is accepting.
	\qed
\end{proof}

We say that a vertex $\pair {q'} {l+1}$ is a \emph{maximal successor} of $\pair q l$
if its rank is maximal amongst all successors of $\pair q l$,
and a sequence $\pair {q_0} l \pair {q_1} {l+1} \cdots \pair {q_h} {l+h}$
is a \emph{maximal path} if, for any $0 \geq k < h$,
$\pair {q_{k+1}} {l+k}$ is a maximal successor of $\pair {q_k} {l+k-1}$.

We define a predecessor and a floor operation on ordinals.
For an ordinal $\alpha$,
its \emph{predecessor} $\alpha - 1$ is either $\alpha$ itself if $\alpha$ is a limit ordinal,
or $\beta$ if $\alpha = \beta + 1$ for some $\beta$;
its \emph{floor} $\floor \alpha  := \sup _{\lambda < \alpha} \lambda$ is the largest limit ordinal strictly smaller than $\alpha$.
Notice that, for $0 < \alpha < \omega^\omega$,
$\floor \alpha < \alpha$.

\begin{lemma}\label{lem:rank_nonincreasing}
	Let vertex $\pair q l$ have rank $\alpha$.
	Then, a) every successor $\pair {q'} {l+1}$ has rank at most $\alpha - 1$,
	and b) there exists a maximal successor attaining rank $\alpha - 1$.
	As a direct consequence,
	c) every node $\pair {q'} {l'}$ reachable from $\pair q l$
	has a smaller rank $\alpha' \leq \alpha$.
\end{lemma}
\begin{proof}
	We split the proof in two cases, depending on whether $\alpha$ is a successor or limit ordinal.
	Let $\alpha$ be a successor ordinal $\beta + 1$.
	Then, $\pair q l$ is a dead end in $G_\beta$,
	and thus it has no successor in $G_\beta$.
	Therefore, each successor $\pair {q'} {l+1}$ has rank $\leq \beta$.
	Moreover, we show that at least one successor has rank exactly equal to $\beta$.
	To this end, let $\beta^* \leq \beta$ be the maximum rank amongst $\pair q l$'s successors.
	Notice that no successor $\pair {q'} {l+1}$ is in $G_{\beta^*}$.
	As $G_\beta \subseteq G_{\beta^*}$,
	it follows that	$\pair q l$ is a dead end in $G_{\beta^*}$.
	Therefore, $\pair q l$ is not in $G_{\beta^* + 1}$,
	which implies it has rank at most $\beta^* + 1 \leq \beta + 1$.
	But $\rank {q, l} = \beta + 1$ by assumption.
	Therefore, $\beta^* = \beta$, as required.
	
	Otherwise, let $\alpha$ be a limit ordinal $\lambda$.
	Thus, $\pair q l$ is inert in $H_\lambda$.
	Let $\pair {q'} {l+1}$ be a successor of $\pair q l$.
	If $\pair {q'} {l+1}$ is not in $H_\lambda$,
	then, since $G_\lambda \subseteq H_\lambda$,
	$\pair {q'} {l+1}$ is not in $G_\lambda$ either.
	Thus, $\pair {q'} {l+1}$ has rank $\leq \lambda$ in this case.
	Otherwise, let $\pair {q'} {l+1}$ be in $H_\lambda$.
	Since $\pair q l$ is inert in $H_\lambda$,
	it follows that $\pair {q'} {l+1}$ is inert in $H_\lambda$ as well.
	Therefore, $\pair {q'} {l+1}$ gets rank exactly equal to $\lambda$ in this case.
	Finally, since $H_\lambda$ does not contain dead ends,
	there exists at least one such inert successor in $H_\lambda$.
	\qed
\end{proof}
\begin{lemma}\label{lem:rank_eventually_accepting}
	If a vertex $\pair {q_0} l$ has a successor ordinal rank $\alpha + 1$,
	then there exists a maximal path $\pair {q_0} l \pair {q_1} {l+1} \cdots \pair {q_h} {l+h}$
	ending in $\pair {q_h} {l+h}$ of rank $\lambda + 1$
	with $\floor{\alpha + 1} \leq \lambda$.
\end{lemma}
\begin{proof}
	We proceed by ordinal induction.
	If $\alpha$ is a limit ordinal $\lambda$,
	the claim holds immediately:
	Take $h = 0$; clearly, $\lambda = \floor{\lambda + 1}$.
	
	Otherwise, let $\alpha$ be a successor ordinal $\beta + 1$.
	That is, vertex $\pair {q_0} l$ has rank $\alpha + 1 = (\beta + 1) + 1$.
	By Lemma~\ref{lem:rank_nonincreasing} b),
	$\pair {q_0} l$ has a maximal successor $\pair {q_1} {l+1}$ of rank $\beta + 1 = \alpha$.
	By induction, there exists a maximal path $\pair {q_1} {l+1} \cdots \pair {q_h} {l+h}$ with $h > 0$,
	ending in $\pair {q_h} {l+h}$ of rank $\lambda + 1$
	with $\floor{\beta + 1} \leq \lambda$.
	But $\beta = \alpha + 1$,
	thus $\floor{\beta + 1} = \floor{\alpha + 1} \leq \lambda$.
	\qed
\end{proof}

\begin{lemma}\label{lem:rank_limit}
	If a vertex $\pair {q_0} l$ has a nonzero limit ordinal rank $\lambda$,
	then there exists a path $\pair {q_0} l \pair {q_1} {l+1} \cdots \pair {q_h} {l+h}$ with $h \geq 1$
	ending in $\pair {q_h} {l+h}$ of rank $\alpha + 1$
	with $\floor\lambda \leq \alpha$.
\end{lemma}
\begin{proof}
	Let $\pair {q_0} l$ have rank $\lambda > 0$.
	By contradiction,
	assume $\pair {q_0} l$ has no descendant $\pair {q'} {l'}$ of rank $\alpha + 1$ with $\floor\lambda \leq \alpha$.
	That is, all descendants $\pair {q'} {l'}$ of successor ordinal rank $\alpha + 1$ have $\alpha < \floor\lambda$,
	which is the same as $\alpha + 1 < \floor\lambda$.
	By definition, $\pair {q_0} l$ is inert in $H_\lambda \subseteq H_{\floor\lambda}$.
	We show that $\pair {q_0} l$ is inert in $H_{\floor\lambda}$ as well.
	This is a contradiction, since $\lambda$ is nonzero,
	therefore $\pair {q_0} l$ would get rank $\floor\lambda < \lambda$.
	
	To this end, we show that any vertex reachable from $\pair {q_0} l$ in $H_{\floor\lambda}$ is non-accepting.
	For such a vertex $\pair {q'} {l'}$ to be accepting,
	by Lemma~\ref{lem:rank_accepting},
	it is necessary to have successor rank $\alpha + 1 < \floor\lambda$.
	Clearly, $\pair {q'} {l'} \not\in H_{\floor\lambda}$.
	Therefore, $\pair {q_0} l$ is inert in $H_{\floor\lambda}$.	
	\qed
\end{proof}

\newcommand{\myrank}[2]{\mathsf{rank}_{#1}^w({#2})}

\ignore{
\begin{lemma}\label{lem:rank_preserved}
	For an automaton $\NBA Q$, let $q_0, s_0 \in Q$
  	and let $w\in\Sigma^\omega$ be an infinite word.
	For two states $\hat q, \hat s \in Q$ and a level $l \geq 0$,
	if $\myrank {q_0} {\hat q, l} \leq \myrank {s_0} {\hat s, l}$,
	then there exists a (maximal) successor $\pair {\hat s'} {l+1}$ of $\pair {\hat s} l$
	such that, for any successor $\pair {\hat q'} {l+1}$ of $\pair {\hat q} l$,
	$\myrank {q_0} {\hat q', l+1} \leq \myrank {s_0} {\hat s', l+1}$.
\end{lemma}
}

\begin{lemma}\label{lem:ordinal_implies_fx_sim}
	Let $w\in\Sigma^\omega$.
	If $\myrank {q_0} {q_0, 0} \leq \myrank {s_0} {s_0, 0}$, then $q_0 \wdesim w s_0$.
\end{lemma}

\begin{proof}
	Assume $\myrank {q_0} {q_0, 0} \leq \myrank {s_0} {s_0, 0}$.
	We show that Duplicator has a winning strategy in $\fxsimgame w {q_0} {s_0}$.
	For any round $i$,
	let $\conf {q_i} {s_i}$ be the current configuration of the simulation game,
	and	let the rank of Spoiler and Duplicator at round $i$ be
	$\myrank {q_0} {q_i, i}$ and $\myrank {s_0} {s_i, i}$, respectively.
	Intuitively, Duplicator wins by ensuring both a \emph{safety} and a \emph{liveness} condition.
	The safety condition requires Duplicator to always preserve the ordering between ranks.
	I.e., at round $i$,
	$\myrank {q_0} {q_i, i} \leq \myrank {s_0} {s_i, i}$.
	The liveness condition enforces Duplicator to (eventually) visit an accepting state if Spoiler does so.
	
	Duplicator plays in two modes, \emph{normal mode} and \emph{obligation mode}.
	In normal mode Duplicator only enforces the safety condition,
	while in obligation mode Duplicator needs to satisfy the liveness condition,
	while still preserving the safety condition.
	
	In normal mode, we asssume that Duplicator's rank is a limit ordinal,
	and, by Lemma~\ref{lem:rank_nonincreasing},
	Duplicator can preserve the rank by always selecting maximal successors.
	We say that Duplicator \emph{plays maximally} during normal mode.
	The game stays in normal mode as long as Spoiler is not accepting.
	Whenever $q_i \in F$ at round $i$,
	then Duplicator switches to obligation mode.
	%
	Suppose that the current rank of Duplicator at round $i$ is a limit ordinal $\lambda$.
	Since $q_i \in F$,
	by Lemma~\ref{lem:rank_accepting} Spoiler's rank is a successor ordinal $\alpha + 1 < \lambda$.
	W.l.o.g. we assume that Spoiler plays maximally during obligation mode.
	By Lemma~\ref{lem:rank_eventually_accepting},
	there exists a maximal path $\pair {q_i} i \pair {q_{i+1}} {i+1} \cdots \pair {q_j} j$
	s.t. Spoiler's rank at round $j \geq i$ is $\lambda' + 1$.
	A further move by Spoiler extends the previous path to $\pair {q_{j+1}} {j+1}$.
	By Lemma~\ref{lem:rank_nonincreasing} b),
	Spoiler's rank at round $j + 1$ is now $\lambda'$,
	and by part c) of the same lemma, $\lambda' \leq \alpha + 1$.
	By part b), Duplicator can play a maximal path
	$\pair {s_i} i \pair {s_{i+1}} {i+1} \cdots \pair {s_{j+1}} {j+1}$
	s.t. Duplicator's rank at round $j + 1$ is $\lambda$.
	Thus, $\lambda' < \lambda$, which implies $\lambda' \leq \floor\lambda$.
	So, let $\conf {q_{j+1}} {s_{j+1}}$ be the configuration at round $j + 1$.
	By Lemma~\ref{lem:rank_limit}, Duplicator can play a path
	$\pair {s_{j+1}} {j+1} \pair {s_{j+2}} {j+2} \cdots \pair {s_k} k$ with $k > j + 1$
	and s.t. Duplicator's rank at round $k$ is $\alpha' + 1$ with $\floor\lambda \leq \alpha'$.
	Therefore, $\lambda' \leq \alpha'$.
	By Lemma~\ref{lem:rank_eventually_accepting},
	Duplicator can extend the previous path with a maximal path
	$\pair {s_k} k \pair {s_{k+1}} {k+1} \cdots \pair {s_h} h$
	s.t. Duplicator's rank at round $h > k$ is $\lambda'' + 1$
	with $\floor{\alpha'+1} \leq \lambda''$.
	By Lemma~\ref{lem:rank_accepting},
	$s_h \in F$, thus Duplicator has satisfied the pending obligation.
	At round $h + 1$, Duplicator's rank is $\lambda''$ by Lemma~\ref{lem:rank_nonincreasing} b),
	and the game can switch to normal mode.
	Notice that $\lambda' \leq \alpha' < \alpha' + 1$ implies $\lambda' \leq \floor{\alpha + 1}$.
	Therefore, $\lambda' \leq \lambda''$ and the safety condition is satisfied.
	\qed
\end{proof}

\begin{lemma}\label{lem:nfx_sim_implies_ordinal}
	Let $w\in\Sigma^\omega$ and $k \geq 1$.
	If $q_0 \wkdesim w k s_0$, then $\myrank {q_0} {q_0, 0} \leq \myrank {s_0} {s_0, 0}$
\end{lemma}

\begin{proof}
	We prove the contrapositive.
	Assume $\myrank {q_0} {q_0, 0} \not\leq \myrank {s_0} {s_0, 0}$.
	Since ordinals are linearly ordered,
	this means $\myrank {q_0} {q_0, 0} > \myrank {s_0} {s_0, 0}$.
	We have to show $q_0 \not\wkdesim w k s_0$, for arbitrary $k \geq 1$.
	Take $n$ to be the size of the automaton.
	We actually prove that Duplicator does not win even with $n$ pebbles,
	i.e., $q_0 \not\wkdesim w n s_0$.

	For any  round $i$,
	let $\conf {q_i} {\set s_i}$ be the current configuration of the simulation game $\kfxsimgame n w {q_0} {s_0}$.
	(For simplicity, we omit the third component.)
	Notice that $\set s_i$ identifies a \emph{subset} of vertices at level $i$ in $G^w_{s_0}$:
	$\set s_i \subseteq \{ s \st \pair s i \in G^w_{s_0} \}$.
	We extend the notion of rank to sets of vertices by taking the maximal rank.
	That is, the rank of Duplicator at round $i$ is $\sup_{s \in \set s_i} \myrank {s_0}  {s, i}$.
	As before, Spoiler's rank is just $\myrank {q_0} {q_i, i}$.
	
	We assume that, at round 0, every pebble has limit rank.
	If not, Spoiler can enforce such a situation by waiting a suitable number of rounds.
	(I.e., by playing maximally according to Lemma~\ref{lem:rank_nonincreasing}.)
	So, let's Spoiler have limit rank $\lambda$ and Duplicator have limit rank $\mu$,
	with $\lambda > \mu$.
	We assume that Duplicator always plays maximally,
	unless she is forced to act differently.
	By Lemma~\ref{lem:rank_limit},
	Spoiler can play a path $\pair {q_0} 0 \pair {q_1} 1 \cdots \pair {q_i} i$ with $i > 0$,
	s.t. her rank at round $i$ is $\alpha + 1$ and $\alpha \geq \floor\lambda$.
	From $\lambda > \mu$ we have $\floor\lambda \geq \mu$,
	which implies $\alpha \geq \mu$.
	By Lemma~\ref{lem:rank_eventually_accepting},
	Spoiler can extend the previous path with a maximal path
	$\pair {q_i} i \pair {q_{i+1}} {i+1} \cdots \pair {q_j} j$ with $j > i$,
	s.t. her rank at round $j$ is $\lambda' + 1$ and $\lambda' \geq \floor{\alpha + 1}$.
	By Lemma~\ref{lem:rank_accepting},
	$q_j \in F$.
	From $\alpha + 1 > \alpha \geq \mu$ we have $\floor{\alpha + 1} \geq \mu$,
	which implies $\lambda' \geq \mu$.
	By performing a further maximal step,
	Spoiler reaches state $\pair {q_{j+1}} {j+1}$,
	thus attaining rank $\lambda'$.
	From now on, Spoiler plays maximally.
	
	Since Duplicator was playing maximally,
	in the meanwhile she replied to Spoiler with a sequence
	$\pair {\set s_0} 0 \pair {\set s_1} 1 \cdots \pair {\set s_{j+1}} {j+1}$
	s.t. she has rank $\mu$ at round $j+1$.
	
	Now, let $\conf {q_{j+1}} {\set s_{j+1}}$ be the current configuration,
	and remember that Duplicator has a pending obligation.
	That is, Duplicator has to ensure that at some future round $k$
	all pebbles are good since round $j+1$.
	Let $\set s_k$ be the position of pebbles at round $k$.
	This implies that every state in $\set s_k$ has an accepting predecessor since round $j+1$.
	By Lemma~\ref{lem:rank_accepting},
	accepting pebbles receive successor ranks,
	and, since ranks are nonincreasing along paths in $G^w_{s_0}$ (by Lemma~\ref{lem:rank_nonincreasing}),
	it follows that every pebble in $\set s_k$ has rank $< \mu$.
	That is, Duplicator's rank at round $k$ is $< \mu$.
	Since Duplicator has now satisfied the pending obligation,
	she will again play maximally, from round $k$ on.
	By Lemma~\ref{lem:rank_nonincreasing},
	all pebbles eventually stabilize to a limit rank.
	Since there is a finite number of pebbles,
	it follows that at some round $h \geq k$ Duplicator's rank is $\mu' < \mu$.
	Let $\set s_h$ be the position of Duplicator's pebbles at round $h$.
	
	In the meanwhile Spoiler replied with a maximal path
	$\pair {q_{j+1}} {j+1} \cdots \pair {q_h} {h}$,
	preserving rank $\lambda' \geq \mu > \mu'$ until round $h$.
	Therefore, $\lambda' > \mu'$ and the situation at round $h$ is identical to the initial situation at round 0.
	
	Since ordinals are well-founded,
	Spoiler can iterate the whole procedure
	and after a finite number of repetitions Duplicator hits the trap rank $\omega$.
	At that point, Spoiler would have a limit rank $\lambda'' \!>\! \omega$,
	so she will just force one more obligation,
	which would remain unmet
	(vertices of rank $\omega$ have no accepting successor).
	Thus, Spoiler wins.
	\qed
\end{proof}

\appendixstatement{Theorem}{thm:fx_sim_equiv}{
	For any NBA $\NBA Q$, $k \geq 1$ and states $q,s \in Q$,
	$q \fxkdesim k s$ iff $q \fxdesim s$.
}
\begin{proof}
	By combining the previous two lemmas, we get
	\begin{align*}
		q \fxdesim s \implies q \fxkdesim k s \implies \left( \myrank {q} {q, 0} \leq \myrank {s} {s, 0} \right) \implies q \fxdesim s \ ,
	\end{align*}
	where the first implication holds by the definition of $\fxkdesim k$,
	and the last two by Lemmas~\ref{lem:nfx_sim_implies_ordinal} and \ref{lem:ordinal_implies_fx_sim},
	respectively.
\end{proof}

\section{Proofs for Section~\ref{sec:correctness}}

\appendixstatement{Lemma}{lem:increasing_seq_invariant_under_subseq}{
	Let $w \in \Sigma^\omega$ and $\pi_0, \pi_1, \dots$ as in Definition~\ref{def:coherent_sequences}.
	If $\Pi = \pi_0, \pi_1, \dots$ is coherent,
	then any infinite subsequence $\Pi' = \pi_{f(0)}, \pi_{f(1)}, \dots$ thereof is coherent.
}

\begin{proof}
	Let $\Pi := \pi_0, \pi_1, \dots$ be an infinite coherent sequence,
	and let $\Pi' := \pi_{f(0)}, \pi_{f(1)}, \dots$ be any infinite subsequence thereof,
	for some $f : \Nat \mapsto \Nat$ with $f(0) < f(1) < \cdots$.
	We have to show
	\begin{align*}
		\forall i' \cdot \exists j' \cdot \exists h' \cdot \forall k' \geq h' \cdot j' < |\pi_{f(k')}| \wedge \countf{\pi_{f(k')}}{j'} \geq i' \ .
	\end{align*}
	Let $i' \in \Nat$.
	By taking $i := i'$, by the coherence of $\Pi$,
	there exists $j, h$ s.t
	\begin{align*}
		(*)\ \forall k \geq h \cdot j < |\pi_k| \wedge \countf{\pi_k}{j} \geq i'.
	\end{align*}
	Let $h'$ be the minimal $m$ s.t. $f(m) \geq h$.
	For any $k' \geq h'$,
	we have $f(k') \geq f(h') \geq h$.
	Thus, by letting $k := f(k')$ in $(*)$,
	we obtain $j < |\pi_{f(k')}| \wedge \countf{\pi_{f(k')}}{j} \geq i'$.
	Take $j' := j$.
	Since $k'\geq h'$ was arbitrary,
	we have proved that $\Pi'$ is coherent.
	\qed
\end{proof}

\appendixstatement{Lemma}{lem:increasing_seq_of_paths}{
	For $w \in \Sigma^\omega$,
	let $\Pi = \pi_0, \pi_1, \dots$ be a coherent sequence of paths over (prefixes of) $w$.
	Then, there exists a fair path $\rho$ over $w$.
	Moreover, if all $\pi_i$'s are initial,
	then $\rho$ is initial.
}
\begin{proof}
	Let $\Pi := \pi_0, \pi_1, \dots$ be a coherent sequence.
	We prove by induction the following claim:
	For $l\in\Nat$,
	$R(l)$ holds iff there exists a finite sequence of finite paths
	$\rho_0 \prefixle \rho_1 \prefixle \dots \prefixle \rho_l$,
	with $\rho_l$ of length $m_l$,
	and an infinite subsequence $\Pi_l := \pi_{f_l(0)}, \pi_{f_l(1)}, \dots$ of $\Pi$
	with $f_l(0) < f_l(1) < \cdots$,
	such that
	\begin{align}
		(a)\ \countf{\rho_l}{m_l} \geq l &&
%
		(b)\ \Pi_l \textrm{ is coherent } &&
%
		(c)\ \forall k \cdot \rho_l \prefixleq \pi_{f_l(k)} \ .
	\end{align}
	
	For the base case $l = 0$,
	take $\rho_0 := \varepsilon$ of length $m_0 := 0$,
	and $f_0(i) = i$ for any $i$.
	Then, $\Pi_0 = \Pi$ and $R(0)$ holds.
	
	For the inductive step,
	assume $R(l-1)$ holds.
	That is, there exist $\rho_0 \prefixle \rho_1 \prefixle \dots \prefixle \rho_{l-1}$,
	with $\rho_{l-1}$ of length $m_{l-1}$,
	and $\Pi_{l-1} = \pi_{f_{l-1}(0)}, \pi_{f_{l-1}(1)}, \dots$
	with $\rho_{l-1} \prefixleq \pi_{f_{l-1}(k)}$ for any $k$.
	Since $\Pi_{l-1}$ is coherent,
	by taking $i := l$,
	there exist $j$ and $h$ s.t.,
	for any $\pi$ in the sequence $\pi_{f_{l-1}(h)}, \pi_{f_{l-1}(h+1)}, \dots$,
	$\pi$ has length at least $j$ and
	$\countf{\pi}{j} \geq l$.
	Since the various $\pi$'s are branches in a finitely-branching tree,
	it follows that at any fixed depth $d$ there are only finitely many different branches of length $d$.
	Therefore, there exists a least one such finite branch which is shared by infinitely many $\pi$'s.
	For $d = j$, we get that there exists a finite path $\rho'$ of length $j$
	s.t. $\countf{\rho'}{j} \geq l$ and $\rho' \prefixleq \pi$ for infinitely many such $\pi$'s.
	Let	$\Pi_l := \pi_{g(f_{l-1}(h))}, \pi_{g(f_{l-1}(h+1))}, \dots$ be this infinite subsequence.
	We assume w.l.o.g. that $m_{l-1} < j$,
	and, consequently, $\rho_l \prefixle \rho'$.
	Take $f_l := g \circ f_{l-1}$,
	$\rho_l := \rho'$ and $m_l := j$.
	Then, (a) and (c) are satisfied by construction,
	while (b) follows by Lemma~\ref{lem:increasing_seq_invariant_under_subseq}.
	This proves $R(l)$,
	concluding the inductive step.
	
	Therefore, one can build the infinite sequence of finite paths
	$\varepsilon = \rho_0 \prefixle \rho_1 \prefixle \cdots$
	such that, for any $l$, $\rho_l$ visits at least $l$ final states.
	Take $\rho$ to be the limit of the $\rho_l$'s.
	Finally, since $\rho_1 \prefixle \pi_{f_1(0)}$ by property $(c)$,
	it follows that if all $\pi_i$'s are initial,
	then so is $\pi_{f_1(0)}$,
	and thus $\rho$.
	\qed	
\end{proof}

\appendixstatement{Theorem}{thm:Upsilon_is_GFQ}{
	Let $R$ be a jumping-safe preorder.
	Then, $R$ is good for quotienting.
}
\begin{proof}
	Assume $R$ is jumping-safe
	and let $\approx_R$ be the equivalence induced by $R$.
	We have to show $\omegalang{\NBA Q} = \omegalang{\quot{\NBA{Q}}{\approx_R}}$.
	The direction $\omegalang{\NBA Q} \subseteq \omegalang{\quot{\NBA{Q}}{\approx_R}}$ holds by Lemma~\ref{lem:quot_contains}.
	
	For the other direction, assume $w \in \omegalang{\quot{\NBA{Q}}{\approx_R}}$,
	with $w = a_0a_1\cdots \in \Sigma^\omega$.
	Let $\pi_{\approx_R} = [q_0] \goesto {a_0} [q_1] \goesto {a_1} [q_2] \cdots$
	be an accepting run over $w$ in $\quot{\NBA{Q}}{\approx_R}$.
	By the definition of quotient,
	for any $i$,
	there exist states $q_i, q^F_i, \hat q_i \in Q$
	s.t. $q_i\ R\ q^F_i\ R\ \hat q_i$
	and $\hat q_i \goesto {a_i} q_{i+1}$.
	That is, $\pi_{\approx_R}$ induces a jumping path $\pi$ as in Equation~\ref{eq:jumping_path}.
	Moreover, $q^F_i$ can be taken in $F$ if $[q_i]$ is accepting.
	Since $[q_0]$ is initial, we assume w.l.o.g. that $q_0 \in I$.
	Since $R$ is jumping-safe and $\pi$ is both initial and fair,
	there exists a coherent sequence of initial paths $\pi_0, \pi_1, \dots$ over prefixes of $w$.
	By Lemma~\ref{lem:increasing_seq_of_paths},
	there exist an (non-jumping) accepting run over $w$ in $Q$.
	Therefore, $w \in \omegalang{\NBA Q}$.
	\qed
\end{proof}

\section{Proofs for Section~\ref{sec:transformers}}

\appendixstatement{Lemma}{lem:tau0_summary}{
	For a preorder $R$, $R \subseteq R \circ \tau_0(R) \subseteq \tau_0(R)$.
}
\begin{proof}
	Directly from Lemmas~\ref{lem:tau0_increasing} and \ref{lem:tau0_closed} below.
\end{proof}

%
\begin{lemma}\label{lem:tau0_increasing}
	For any reflexive $R$,
	$R \subseteq \tau_0(R)$.
\end{lemma}
\begin{proof}
	Let $T := \tau_0(R)$,
	and assume $s\ R\ q$.
	We have to show $s\ T\ q$.
	Let's Spoiler select $a$ and $q'$ s.t.
	$q \goesto a q'$.
	Since $s\ R\ q$ by assumption,
	Duplicator can directly take $\hat s := q$.
	Trivially $q \in F \implies \hat s \in F$,
	as required by the winning condition.
	\qed
\end{proof}
%
%
\begin{lemma}\label{lem:tau0_closed}
	For any transitive $R$,
	$R \circ \tau_0(R) \subseteq \tau_0(R)$.
\end{lemma}
\begin{proof}
	Let $T := \tau_0(R)$,
	and assume $\bar s\ R\ s\ T\ q$.
	We have to show $\bar s\ T\ q$.
	Let's Spoiler select $a$ and $q'$ s.t.
	$q \goesto a q'$.
	Since $s\ T\ q$ by assumption,
	Duplicator can select $\hat s$ s.t. $s\ R\ \hat s$
	and $\hat s \goesto a s'$, for some $s'$.
	Then, by transitivity, $\bar s\ R\ \hat s$.
	As $q \in F \implies \hat s \in F$ (by $s\ T\ q$),
	we conclude that Duplicator wins from $\bar s$ as well,
	thus $\bar s\ T\ q$.
	\qed
\end{proof}
\ignore{ 
For transitive $R$,
$\tau(R)$ is closed w.r.t. $R$ both to the left and to the right.
\begin{lemma}\label{lem:tau_closed}
	For any transitive $R$,
	$R \circ \tau(R) \circ R \subseteq \tau(R)$.
\end{lemma}
\begin{proof}
	Let $T := \tau(R)$.
	By Lemma~\ref{lem:tau_increasing},
	we have $R \subseteq T$.
	Thus, 
	\begin{align*}
		R \circ T \circ R \subseteq T \circ T \circ T \subseteq T \ ,
	\end{align*}
	where the second inclusion follows from the transitivity of $T$ (see Lemma~\ref{lem:tau_transitive}).
\end{proof}
}
%
%
\appendixstatement{Theorem}{thm:tau0_GFQ}{
	Let $R$ a $F$-respecting preorder,
	and let $T \subseteq \tau_0(R)$ be an appealing, improving fragment of $\tau_0(R)$.
	If $R$ is jumping-safe, then $T$ is jumping-safe.
}
\begin{proof}
	Assume that $R$ is jumping-safe and $F$-respecting,
	and let $T$ be an appealing, improving fragment of $\tau_0(R)$.
	That is, $T$ is a self-respecting and transitive fragment of $\tau_0(R)$,
	with $R \subseteq T$.
	We have to show that $T$ is jumping-safe.
	To this end, let $w = a_0a_1\cdots \in \Sigma^\omega$, and let the following be an initial $T$-jumping path
	\begin{align*}
		\pi = q_0\ T\ q^F_0\ T\ \hat q_0 \goesto {a_0} q_1\ T\ q^F_1\ T\ \hat q_1 \goesto {a_1} q_2 \cdots ,\quad q_0 \in I \ .
	\end{align*}
	
	First, we show by induction the following claim:
	For any $i \geq 0$, there exists a finite initial path
	\begin{align*}
		\rho_i = r_0\ R\ \hat r_0 \goesto {a_0} r_1\ R\ \hat r_1 \goesto {a_1} \cdots r_i, \quad r_0 \in I \ ,
	\end{align*}
	s.t. $r_i\ T\ q_i$, and,
	for any $0 \leq k < i$,
	$q^F_k \in F \implies \hat r_k \in F$.

	For $i = 0$, just take $r_0 := q_0$.
	For $i \geq 0$, assume $\rho_i = r_0\ R\ \hat r_0 \goesto {a_0} r_1 \cdots r_i$ has already been built.
	Since $q^F_i\ T\ \hat q_i \goesto {a_i} q_{i+1}$,
	by the definition of $T$ there exists $\hat q^F_i \goesto {a_i} q'$
	for some $\hat q^F_i$ and $q'$ with $q^F_i\ R\ \hat q^F_i$ and $q'\ T\ q_{i+1}$.
	But $q_i\ T\ q^F_i$ and, by induction hypothesis, $r_i\ T\ q_i$.
	Since $T$ is transitive,
	we get $r_i\ T\ q^F_i$,
	so there exists $\hat r_i \goesto {a_i} r_{i+1}$
	with $r_i\ R\ \hat r_i$ and $r_{i+1}\ T\ q'$.
	Again by transitivity, we get $r_{i+1}\ T\ q_{i+1}$.
	Moreover, if $q^F_i \in F$, then since $R$ respects final states,
	we have $\hat q^F_i \in F$,
	and, by the definition of $T$,
	we finally derive $\hat r_i \in F$.
	Thus, we have just built $\rho_{i+1} = r_0\ R\ \hat r_0 \goesto {a_0} r_1 \cdots r_i\ R\ \hat r_i \goesto {a_i} r_{i+1}$.
	This concludes the inductive step,
	and the claim is proved.
	
	From the claim above,
	let $\rho$ be the infinite initial $R$-jumping sequence resulting by taking limit of the $\rho_i$'s.
	Since $R$ is jumping-safe,
	there exists an infinite sequence of initial finite paths
	$\pi_0, \pi_1, \dots$ s.t. $\last{\pi_i}\ R\ r_i$.
	By assumption $R \subseteq T$,
	so $\last{\pi_i}\ T\ r_i$ holds as well.
	By $r_i\ T\ q_i$ and transitivity,
	we obtain $\last{\pi_i}\ T\ q_i$.
	Therefore, the same sequence $\pi_0, \pi_1, \dots$ can be taken as a witness for $T$ being jumping-safe.
	
	Finally, assume that $\pi$ is fair,
	i.e., $q^F_i \in F$ for infinitely many $i$'s.
	By the claim above, $\hat r_i \in F$ for infinitely many $i$'s,
	therefore $\rho$ is fair as well.
	Since $R$ is jumping-safe
	(by taking $r^F_i := \hat r_i$, $R$ being reflexive),
	we finally infer that $\pi_0, \pi_1, \dots$ is coherent,
	which concludes the proof.
	\qed
\end{proof}

\appendixstatement{Lemma}{lem:tau0_notimproving}{
	For any reflexive $R$, let $T \subseteq \tau_0(R)$ be any appealing fragment of $\tau_0(R)$.
	Then, $\tau_0(T) \subseteq \tau_0(R)$.
	That is, at the second iteration $\tau_0$ does not introduce any new fragment which could not be found before.
}
\begin{proof}
	Let $R$ be reflexive.
	Let $T$ be an appealing (= transitive and self-respecting) fragment of $V_0 := \tau_0(R)$,
	and let $V_1 := \tau_0(T)$.
	We have to show $V_1 \subseteq V_0$.
	To this end, let $s\ V_1\ q$ 
	and let Spoiler choose a transition $q \goesto a q'$.
	By the definition of $V_1$,
	there exist $s\ T\ \bar s$ and $\bar s'$
	with $\bar s \goesto a \bar s'$ and $\bar s'\ V_1\ q'$.
	By the definition of $T$,
	there exist $s\ R\ \hat s$ and $s'$
	with $\hat s \goesto a s'$ and $s'\ T\ \bar s'$ (since $T$ is self-respecting).
	$T$ being transitive, from $s'\ V_1\ \bar s'\ T\ q'$ and from Lemma~\ref{lem:tau0_closed},
	we get $s'\ V_1\ q'$.
	Thus, we let Duplicator choose $\hat s$ and $s'$ above,
	as required by the definition of $V_0$.
	Duplicator is winning as $q \in F$ implies $\bar s \in F$,
	and the latter implies $\hat s \in F$,
	the first implication holding by the definition of $V_1$,
	and the second by $T$.
	Therefore, $s\ V_0\ q$.
	\qed
\end{proof}

\begin{lemma}\label{lem:tau1_transitive}
	For any relation $R$, $\tau_1(R)$ is transitive.
\end{lemma}
\begin{proof}
	Let $T := \tau_1(R)$, and let $s\ T\ r\ T\ p$.
	We have to show $s\ T\ p$.
	Let Spoiler choose $a$ and $\hat p$ and $p'$ s.t.
	$p\ R\ \hat p$ and $\hat p \goesto a p'$.
	We have to show
	1) that Duplicator can choose $\hat s$ and $s'$ s.t.
	$s\ R\ \hat s$ and $\hat s \goesto a s'$,
	and 2) $\hat p \in F \implies \hat s \in F$.
	For 1), from $r\ T\ p$ it follows that there exist $\hat r$ and $r'$ s.t.
	$r\ R\ \hat r$ and $\hat r \goesto a r'$.
	Then, from $s\ T\ r$ one can directly find the required $\hat s$ and $s'$.
	For 2), assume $\hat p \in F$.
	From $r\ T\ p$ it follows that the $\hat r$ found above is in $F$ as well.
	Finally, $\hat s \in F$ follows from $s\ T\ r$ in a similar way.
	\qed
\end{proof}
\begin{lemma}\label{lem:tau1_increasing}
	For any transitive $R$,
	$R \subseteq \tau_1(R)$.
\end{lemma}
\begin{proof}
	Let $T := \tau_1(R)$,
	and assume $s\ R\ q$.
	We have to show $s\ T\ q$.
	Let's Spoiler select $a$ and $\hat q$ and $q'$ s.t.
	$q\ R\ \hat q$ and $\hat q \goesto a q'$.
	Since $s\ R\ q$ by assumption,
	and from $R$ being transitive,
	we have $s\ R\ \hat q$.
	Thus Duplicator can directly take $\hat s := \hat q$.
	Finally, trivially $\hat q \in F \implies \hat s \in F$,
	as required by the winning condition.
	\qed
\end{proof}

\appendixstatement{Lemma}{lem:tau1_maximal}{
	For any $R$,
	let $T \subseteq \tau_0(R)$ be any appealing fragment of $\tau_0(R)$.
	If $R \subseteq T$ (i.e., $R$ is improving),
	then $T \subseteq \tau_1(R)$.
}
\begin{proof}
	Let $R, T$ as in the statement of the lemma,
	and let $V := \tau_1(R)$.
	We have to show $T \subseteq V$.
	Let $q\ T\ p$,
	and let Spoiler choose $\hat p$ and $p'$
	with $p\ R\ \hat p$ and $\hat p \goesto a p'$,
	as required by the definition of $V$.
	Then, as $R \subseteq T$ by assumption,
	and $T$ being transitive,
	we have $q\ T\ \hat p$.
	Therefore, by the definition of $T$,
	Duplicator can choose $\hat q$ and $q'$
	with $q\ R\ \hat q$ and $\hat q \goesto a q'$.
	Since $T$ is self-respecting,
	we have $p'\ T\ q'$.
	Finally, $\hat p \in F \implies \hat s \in F$ by the definition of $T$.
	Therefore, Duplicator is winning, and $q\ V\ p$.
	\qed
\end{proof}

\begin{figure}
	
	\centering
	
		\VCDraw{
			\begin{VCPicture}{(0,-1)(9,4)}


			\FinalState[q_1]{(1.5,4)}{Q1}
			\State[q_0]{(4.5,4)}{Q0} \Initial[ne]{Q0}

			\State[q_2]{(0,2)}{Q2}
			\State[q_3]{(3,2)}{Q3}
			\FinalState[q_4]{(6,2)}{Q4}

			\FinalState[q_5]{(1.5,0)}{Q5}
			\State[q_6]{(4.5,0)}{Q6}

			
			\EdgeR{Q0}{Q1}{a}
			\EdgeR{Q0}{Q2}{a}
			\EdgeR{Q0}{Q3}{a}
			\EdgeR[0.3]{Q0}{Q4}{a}
			
			\EdgeR{Q1}{Q2}{a}
			
			\LoopNW[0.3]{Q2}{a}
			\EdgeL[0.6]{Q2}{Q3}{a}
			
			\EdgeR{Q3}{Q5}{a}
			
			\EdgeR{Q4}{Q3}{a}

			\EdgeL{Q5}{Q6}{a}
			
			\EdgeR{Q4}{Q3}{a}
			\LoopNE[0.6]{Q6}{a}
			
			\end{VCPicture}
		}
	
	\caption{
	Quotienting w.r.t. appealing fragments of $\taudenaught$ is incorrect,
	already for unary automata.
	We have $q_3 \bwdisim q_2$ and $q_4 \bwdisim q_2$,
	and the relation $T := \{(q_i, q_i) \st 0 \leq i \leq 6 \} \cup \{ (q_i, q_6) \st 0 \leq i \leq 6 \} \cup \{ (q_i, q_j) \st 2 \leq i,j \leq 4 \} \cup \{ (q_i, q_5) \st 2 \leq i \leq 4\}$ is an appealing fragment of $\taudenaught(\bwdisim)$.
	(In particular, $q_3\ \taudenaught(\bwdisim)\ q_4$ since $q_3$ can ``jump'' to $q_2$.)
	The equivalence induced by $T$ identifies the states $q_2, q_3, q_4$,
	but this is incorrect as the resulting automaton would accept the spurious word $a^\omega$.
	}
	
	\label{fig:taudenaught_not_GFQ}
	
\end{figure}

\appendixstatement{Lemma}{lem:taude_transitive}{
	For any $R$,
	$\taude(R)$ is transitive.
}
\begin{proof}
	A complete and formal proof of transitivity requires the machinery of logbooks and composition of (winning) strategies,
	which is a standard tool for delayed simulation
	(for more details see, e.g., \cite{wilke:fritz:simulations:05}). 
	Here, we highlight the ingredients pertinent to $\taude$.
	
	Let $T := \taude(R)$,
	and let $r\ T\ q\ T\ p$.
	We have to show $r\ T\ p$.
	Let $G_0$ be the game between $r$ and $q$,
	let $G_1$ be the game between $q$ and $p$,
	and let $G$ be the outer game between $r$ and $p$.

	The idea is that Duplicator plays $G$ and at the same time updates $G_0, G_1$ accordingly.
	At round $i$, if the $G$-configuration is $\conf {r_i} {p_i}$,
	then there exists $q_i$ s.t. the $G_0$ configuration is $\conf {r_i} {q_i}$
	and the $G_1$ configuration is $\conf {q_i} {p_i}$.
	
	Let Spoiler choose $\hat p$ and transition $\hat p \goesto {a_i} p_{i+1}$,
	with $p_i\ R\ \hat p$.
	Since $G_1$-Duplicator is winning,
	there exist $\hat q$ and transition $\hat q \goesto {a_i} q_{i+1}$,
	with $q_i\ R\ \hat q$.
	Similarly, $G_0$-since Duplicator is winning,
	there exist $\hat r$ and transition $\hat r \goesto {a_i} r_{i+1}$,
	with $r_i\ R\ \hat r$.
	Thus, Duplicator can proceed in $G$ by taking the last transition above.
	The configuratons are updated as follows:
	The game $G_0$ goes to $\conf {r_{i+1}} {q_{i+1}}$,
	$G_1$ goes to $\conf {q_{i+1}} {p_{i+1}}$
	and $G$ goes to $\conf {r_{i+1}} {p_{i+1}}$.
	
	We now argue that the strategy above is winning.
	W.l.o.g. we assume that the games $G_0, G_1$ are updated according to a fixed winning strategy.
	We show that Duplicator is winning in $G$.
	Assume $\hat p_i \in F$.
	Since $G_1$-Duplicator is playing according to a winning strategy,
	there exists $k \geq i$ s.t. $\hat q_k \in F$.
	Similarly, as $G_0$-Duplicator is playing according to a winning strategy,
	there exists $j \geq k \geq i$ s.t. $\hat r_j \in F$.
	Thus, take $j \geq i$ s.t. $\hat r_j \in F$, as required.
	\qed
\end{proof}

\begin{lemma}\label{lem:taude_increasing}
	For any transitive $R$,
	$R \subseteq \taude(R)$.
\end{lemma}
\begin{proof}
	Immediate from $R \subseteq \tau_1(R)$ by Lemma~\ref{lem:tau1_increasing},
	and $\tau_1(R) \subseteq \taude(R)$ by definition.
	\qed
\end{proof}

\appendixstatement{Theorem}{thm:taude_GFQ}{
	If $R$ is a jumping-safe $F$-respecting preorder,
	then $\taude(R)$ is jumping-safe.
}
\begin{proof}
	%
	Assume that $R$ is a jumping-safe, $F$-respecting preorder,
	and let $T := \taude(R)$.
	We have to show that $T$ is jumping-safe.
	During the proof we refer to Figure~\ref{fig:fxde_transformer_proof_idea},
	hereafter called ``the diagram''.
	Let $w = a_0a_1\cdots \in \Sigma^\omega$,
	and let $\pi$ be an initial $T$-jumping path
	\begin{align*}
			\pi = q_0\ T\ q^F_0\ T\ \hat q_0 \goesto {a_0} q_1\ T\ q^F_1\ T\ \hat q_1 \goesto {a_1} q_2 \cdots ,\quad q_0 \in I \ .
	\end{align*}
	See the blue path in the diagram.
	We inductively show how to build the rest of the diagram,
	and then we use this construction for showing that $T$ is jumping-safe.
	
	Formally, we inductively build a sequence
	$\rho_0, \rho_1, \dots, \rho_i$
	such that, for any $k \leq i$,
	$\rho_k$ is a $T$-ordered $k+4$-tuple of states representing the $k$-th layer of the diagram,
	\begin{align*}
		\rho_k = s_k^0\ T\ s_k^1\ T\ \cdots\ T\ s_k^{k-1}\ T\ s_k^k\ T\ q_k\ T\ q_k^F\ T\ \hat q_k \ .
	\end{align*}
	%
	Two successive layers are in relations with transitions as follows (cf. the diagram):
	\begin{align*}
		\forall \left(1 \leq h \leq k\right) \cdot s_k^h \jgoesto {a_k} s_{k+1}^h, \quad
		q_k^F \jgoesto {a_k} s_{k+1}^{k+1}, \quad
		\hat q_k \goesto {a_k} q_{k+1} \ ,
	\end{align*}
	where the dashed arrow $x \jgoesto a y$ represents an $R$-jumping transition via some suitable proxy.
	That is, $x \jgoesto a y$ iff there exists a proxy $\hat x$ s.t. $x\ R\ \hat x$ and $\hat x \goesto a y$.
	
	For $i = 0$, just take $s_0^0 := q_0$.
	Then, the invariant is clearly satisfied,
	as $q_0\ T\ q^F_0$ by assumption
	and $q_0 \T\ q_0$ by $T$ being reflexive.
	
	For $i \geq 0$,	assume
	$\rho_0, \rho_1, \dots, \rho_i$ has already been built.
	By induction hypothesis,
	$\rho_i$ is the following $T$-ordered tuple:
	\begin{align*}
		\rho_i = s_i^0\ T\ s_i^1\ T\ \cdots\ T\ s_i^{i-1}\ T\ s_i^i\ T\ q_i\ T\ q_i^F\ T\ \hat q_i \ ,
	\end{align*}
	The next layer $\rho_{i+1}$,
	\begin{align*}
		\rho_{i+1} = s_{i+1}^0\ T\ s_{i+1}^1\ T\ \cdots\ T\ s_{i+1}^{i-1}\ T\ s_{i+1}^i\ T\ s_{i+1}^{i+1}\ T\ q_{i+1}\ T\ q_{i+1}^F\ T\ \hat q_{i+1} \ ,
	\end{align*}
	is obtained as follows.
	The last three components $q_{i+1}, q_{i+1}^F,\hat q_{i+1}$ are fixed by the $T$-jumping path $\pi$.
	The rest is determined next.
	Since $\hat q_i \goesto {a_i} q_{i+1}$,
	we propagate the transition down the chain,
	by using the definition of $T$---%
	as indicated by the zigzag arrows in the diagram.
 	As $q^F_i\ T\ \hat q_i$,
	there exists an $R$-jumping transition $q^F_i \jgoesto {a_i} q'\ T\ q_{i+1}$.
	Take $s_{i+1}^{i+1} := q'$.
	Similarly, from $s_i^i\ T\ q_i\ T\ q^F_i$
	there exists $s_i^i \jgoesto {a_i} q''\ T\ s_{i+1}^{i+1}$.
	Take $s_{i+1}^i := q''$.
	Clearly, one can build all the remaining states down to $s_{i+1}^0$	in the same way,
	thus completing layer $i+1$ in the diagram.
	This concludes the inductive step in the definition of $\rho_{i+1}$.
	
	\begin{unremark}
		We assume that each time a new $T$-game starts from configuration $\pair {q_i^F} {q_i}$,
		Duplicator fixes a winning strategy,
		and alway plays accordingly. 
	\end{unremark}
	
	We now prove that final states are ``propagated'' in the diagram right-to-left, top-to-bottom:
	Formally, we show that,
	for any $i \geq 0$,
	if $q^F_i \in F$,
	then there exists $j \geq i$ s.t. $\hat s_j^0 \in F$,
	where $\hat s_j^0$ is the proxy witnessing $s_j^0 \jgoesto {a_j} s_{j+1}^0$.
	Assume $q^F_i \in F$.
	Then, since $R$ is $F$-respecting,
	$\hat q^F_i \in F$,
	where $\hat q^F_i$ is the proxy witnessing $q^F_i \jgoesto {a_i} s_{i+1}^{i+1}$.
	Since $s_i^i\ T\ q^F_i$, by the definition of $\taude$ and by the above remark,
	there exists $j_0 \geq i$ s.t. $\hat s_{j_0}^i \in F$,
	where $\hat s_{j_0}^i$ is the proxy witnessing $s_{j_0}^i \jgoesto {a_{j_0}} s_{{j_0}+1}^i$.
	But $s_{j_0}^{i-1}\ T\ s_{j_0}^i$,
	therefore there exists $j_1 \geq j_0$ s.t. $\hat s_{j_1}^{i-1} \in F$, and so on \dots
	until we reach index $j_i \geq j_{i-1}$,
	for which $\hat s_{j_i}^0 \in F$.
	Thus, take $j := j_i$.
	
	We are finally ready to prove that $T$ is jumping-safe.
	Notice that the leftmost path in the diagram represents an initial $R$-jumping path $\pi'$,
	\begin{align*}
		\pi' = s_0^0\ R\ \hat s_0^0 \goesto {a_0} s_1^0\ R\ \hat s_1^0 \goesto {a_1} \cdots,\ s_0^0 = q_0 \in I \ .
	\end{align*}
	Since $R$ is jumping-safe,
	there exists an infinite sequence of initial finite paths
	$\pi_0, \pi_1, \dots$ s.t. $\last{\pi_i}\ R\ s_i^0$.
	Since $R$ is transitive,
	$R \subseteq T$ by Lemma~\ref{lem:taude_increasing}.
	Therefore, $\last{\pi_i}\ T\ s_i^0$.
	By $s_i^0\ T\ q_i$ and transitivity,
	we obtain $\last{\pi_i}\ T\ q_i$.
	Therefore, the same sequence $\pi_0, \pi_1, \dots$ can be taken as a witness for $T$ being jumping-safe. 	
	Finally, since $\pi$ is fair, i.e., $q^F_i \in F$ for infinitely many $i$'s,
	then $\pi'$ is fair,
	as final states are ``propagated'' (shown above).
	Since $R$ is jumping-safe,
	we conclude that $\pi_0, \pi_1, \dots$ is coherent.
	\qed
\end{proof}

\begin{figure}
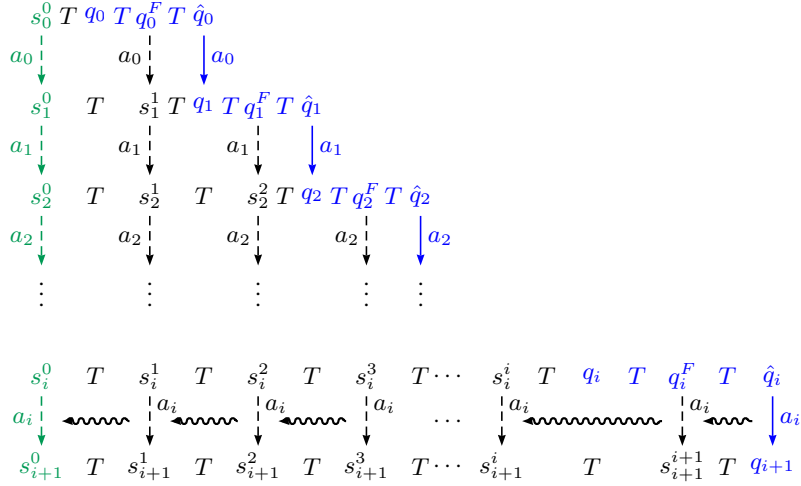

	
	\centering
	\VCDraw{
		\begin{VCPicture}{(0,-1)(16,10)}
			
		\ChgStateLineStyle{none}

		\ChgStateLabelColor{ForestGreen}
		\State[s_0^0]{(0.0,10)}{s00}
		\RstStateLabelColor
		\State[T]{(0.6,10)}{lbl}
		\ChgStateLabelColor{blue}
		\State[q_0]{(1.2,10)}{q0}
		\State[T]{(1.8,10)}{lbl}
		\State[q_0^F]{(2.4,10)}{q0f}
		\State[T]{(3.0,10)}{lbl}
		\State[\hat q_0]{(3.6,10)}{q0hat}
		\RstStateLabelColor

		\ChgStateLabelColor{ForestGreen}
		\State[s_1^0]{(0.0,8)}{s10}
		\RstStateLabelColor
		\State[T]{(1.2,8)}{lbl}
		\State[s_1^1]{(2.4,8)}{s11}
		\State[T]{(3,8)}{lbl}
		\ChgStateLabelColor{blue}
		\State[q_1]{(3.6,8)}{q1}
		\State[T]{(4.2,8)}{lbl}
		\State[q_1^F]{(4.8,8)}{q1f}
		\State[T]{(5.4,8)}{lbl}
		\State[\hat q_1]{(6.0,8)}{q1hat}
		\RstStateLabelColor
		
		\ChgEdgeLineColor{blue}
		\ChgEdgeLabelColor{blue}
		\EdgeL{q0hat}{q1}{a_0}
		\RstEdgeLabelColor
		\RstEdgeLineColor
		\ChgEdgeLineStyle{dashed}
		\EdgeR{q0f}{s11}{a_0}
		\ChgEdgeLineColor{ForestGreen}
		\ChgEdgeLabelColor{ForestGreen}
		\EdgeR{s00}{s10}{a_0}
		\RstEdgeLabelColor
		\RstEdgeLineColor
		\RstEdgeLineStyle

		\ChgStateLabelColor{ForestGreen}
		\State[s_2^0]{(0.0,6)}{s20}
		\RstStateLabelColor
		\State[T]{(1.2,6)}{lbl}
		\State[s_2^1]{(2.4,6)}{s21}
		\State[T]{(3.6,6)}{lbl}
		\State[s_2^2]{(4.8,6)}{s22}
		\State[T]{(5.4,6)}{lbl}
		\ChgStateLabelColor{blue}
		\State[q_2]{(6.0,6)}{q2}
		\State[T]{(6.6,6)}{lbl}
		\State[q_2^F]{(7.2,6)}{q2f}
		\State[T]{(7.8,6)}{lbl}
		\State[\hat q_2]{(8.4,6)}{q2hat}
		\RstStateLabelColor
		
		\ChgEdgeLineColor{blue}
		\ChgEdgeLabelColor{blue}
		\EdgeL{q1hat}{q2}{a_1}
		\RstEdgeLabelColor
		\RstEdgeLineColor
		\ChgEdgeLineStyle{dashed}
		\EdgeR{s11}{s21}{a_1}
		\EdgeR{q1f}{s22}{a_1}
		\ChgEdgeLineColor{ForestGreen}
		\ChgEdgeLabelColor{ForestGreen}
		\EdgeR{s10}{s20}{a_1}
		\RstEdgeLabelColor
		\RstEdgeLineColor
		\RstEdgeLineStyle
		
		\State[\vdots]{(0.0,4)}{s30}
		\State[\vdots]{(2.4,4)}{s31}
		\State[\vdots]{(4.8,4)}{s32}
		\State[\vdots]{(7.2,4)}{s33}
		\State[\vdots]{(8.4,4)}{q3}
		
		\ChgEdgeLineColor{blue}
		\ChgEdgeLabelColor{blue}
		\EdgeL{q2hat}{q3}{a_2}
		\RstEdgeLabelColor
		\RstEdgeLineColor
		\ChgEdgeLineStyle{dashed}
		\EdgeR{s21}{s31}{a_2}
		\EdgeR{s22}{s32}{a_2}
		\EdgeR{q2f}{s33}{a_2}
		\ChgEdgeLineColor{ForestGreen}
		\ChgEdgeLabelColor{ForestGreen}
		\EdgeR{s20}{s30}{a_2}
		\RstEdgeLabelColor
		\RstEdgeLineColor
		\RstEdgeLineStyle

		\ChgStateLabelColor{ForestGreen}
		\State[s_i^0]{(0.0,2)}{si0}
		\RstStateLabelColor
		\State[T]{(1.2,2)}{lbl}
		\State[s_i^1]{(2.4,2)}{si1}
		\State[T]{(3.6,2)}{lbl}
		\State[s_i^2]{(4.8,2)}{si2}
		\State[T]{(6.0,2)}{lbl}
		\State[s_i^3]{(7.2,2)}{si3}
		\State[T]{(8.4,2)}{lbl}
		\State[\cdots]{(9.0,2)}{lbl}
		\State[s_i^i]{(10.2,2)}{sii}
		\State[T]{(11.2,2)}{lbl}
		\ChgStateLabelColor{blue}
		\State[q_i]{(12.2,2)}{qi}
		\State[T]{(13.2,2)}{lbl}
		\State[q_i^F]{(14.2,2)}{qif}
		\State[T]{(15.2,2)}{lbl}
		\State[\hat q_i]{(16.2,2)}{qihat}
		\RstStateLabelColor

		\ChgStateLabelColor{ForestGreen}
		\State[s_{i+1}^0]{(0.0,0)}{sip0}
		\RstStateLabelColor
		\State[T]{(1.2,0)}{lbl}
		\State[s_{i+1}^1]{(2.4,0)}{sip1}
		\State[T]{(3.6,0)}{lbl}
		\State[s_{i+1}^2]{(4.8,0)}{sip2}
		\State[T]{(6.0,0)}{lbl}
		\State[s_{i+1}^3]{(7.2,0)}{sip3}
		\State[T]{(8.4,0)}{lbl}
		\State[\cdots]{(9.0,0)}{lbl}
		\State[s_{i+1}^i]{(10.2,0)}{sipi}
		\State[T]{(12.2,0)}{lbl}
		\State[s_{i+1}^{i+1}]{(14.2,0)}{sipip}
		\State[T]{(15.2,0)}{lbl}
		\ChgStateLabelColor{blue}
		\State[q_{i+1}]{(16.2,0)}{qip}
		\RstStateLabelColor
		
		\ChgEdgeLineColor{blue}
		\ChgEdgeLabelColor{blue}
		\EdgeL{qihat}{qip}{a_i}
		\RstEdgeLabelColor
		\RstEdgeLineColor
		\ChgEdgeLineStyle{dashed}
		\EdgeL[0.2]{si1}{sip1}{a_i}
		\EdgeL[0.2]{si2}{sip2}{a_i}
		\EdgeL[0.2]{si3}{sip3}{a_i}
		\EdgeL[0.2]{sii}{sipi}{a_i}
		\EdgeL[0.2]{qif}{sipip}{a_i}
		\ChgEdgeLineColor{ForestGreen}
		\ChgEdgeLabelColor{ForestGreen}
		\EdgeR{si0}{sip0}{a_i}
		\RstEdgeLabelColor
		\RstEdgeLineColor
		\RstEdgeLineStyle

		\ChgStateFillStatus{none}
		\State[]{(0,1)}{x0}
		\State[]{(2.4,1)}{x1}
		\State[]{(4.8,1)}{x2}
		\State[]{(7.2,1)}{x3}
		\State[\cdots]{(9.0,1)}{lbl}
		\State[]{(10.2,1)}{x4}
		\State[]{(14.2,1)}{x5}
		\State[]{(16.2,1)}{x6}

		\ChgEdgeLineWidth{0.8}
		\ChgZZSize{0.4cm}
		\ZZEdgeL{x6}{x5}{}
		\ZZEdgeL{x5}{x4}{}
		\ZZEdgeL{x3}{x2}{}
		\ZZEdgeL{x2}{x1}{}
		\ZZEdgeL{x1}{x0}{}

		\RstStateLineStyle
		\end{VCPicture}
	}
	\caption{
	Construction for the proof of Theorem~\ref{thm:taude_GFQ}.
	}
	\label{fig:fxde_transformer_proof_idea}
\end{figure}

By using similar techniques,
it is possible to show that repeated application of $\taude$ does not give coarser relations.
This is analogous of what proved in Lemma~\ref{lem:tau0_notimproving} for $\tau_0$.
The proof of this fact is omitted.

\begin{lemma}
	For any preorder $R$, $\taude(\taude(R)) \subseteq \taude(R)$.
\end{lemma}

\ignore{

\begin{proof}[of Lemma~\ref{lem:taude_notimproving}]
	Let $T := \taude(R)$, and assume $s\ \taude(T)\ q$.
	We have to prove $s\ T\ q$.
	
	Since $R$ is transitive, by Lemma~\ref{lem:taude_increasing}
	$R \subseteq T.$
\end{proof}

}

\section{Computing $\taude(R)$}\label{app:computing_taude}

In this section we give an algorithm for computing $\taude(R)$ from Section~\ref{sec:taude},
obtained as an extension of the classical algorithm for computing delayed simulation \cite{etessami:etal:fairsimulations:05}.
We assume that the relation $R$ has already been computed.
We build a game graph where Duplicator has a B\"uchi winning objective.

We enrich configurations from the basic semantic game for $\taude(R)$ with an obligation bit
recording whether Duplicator has to visit an accepting state.
Formally, Spoiler's positions are of the form $\triple s q b$, with $q, s \in Q$ and $b \in \{0,1\}$,
and Duplicator's positions are of the form $\quintuple s q {\hat b} a {q'}$,
with $q, s, q' \in Q$, $a \in \Sigma$ and $\hat b \in \{0,1\}$.
Spoiler can pick a move $(\triple s q b, \quintuple s q {\hat b} a {q'})\in\Gamma_0'$
if there exists $\hat q \in Q$ s.t.
$q R \hat q \goesto a q'$,
and $\hat b = 1$ if $\hat q \in F$ and $b$ otherwise.
Similarly, Duplicator can pick a move $(\quintuple s q {\hat b} a {q'}, \triple {s'} {q'} {b'})\in\Gamma_1'$
if there exists $\hat s \in Q$ s.t.
$s R \hat s \goesto a s'$,
and $b' = 0$ if $\hat s \in F$ and $\hat b$ otherwise.
The objective for Duplicator is to ensure that the winning bit is $0$ infinitely often,
that is, every obligation to visit an accepting state is eventually met.
Formally, the winning condition is
\[W' = \setof{\triple {s_0} {q_0} {b_0} \triple {s_1} {q_1} {b_1} \cdots}{\forall i \geq 0 \cdot \exists j \geq i \cdot b_j = 0} \ .\]

Let $\CPrename$ be a controlled predecessor operator for Duplicator, defined as
\begin{align*}
	\CPre{}{X} = \setof{x}{\forall (x, y) \in \Gamma_0' \cdot \exists (y, z) \in \Gamma_1' \cdot z \in X } \ .
\end{align*}
That is, $x = \triple s q b \in \CPre{}{X}$ if Duplicator can force the game in $X$ in one step from configuration $x$.
Then, the winning region for Duplicator can be computed by evaluating the following fixpoint:
\begin{align*}
	V = \nu X \cdot \mu Y \cdot [b = 0] \cap \CPre{}{X} \cup \CPre{}{Y}\ ,
\end{align*}
where with $[b=0]$ we have indicated the set of configurations with no obligation pending,
i.e., $[b=0] = \setof{\triple q s b}{q, s \in Q, b = 0}$.
Finally, $s\ \taude(R)\ q$ holds iff $\triple s q 0 \in V$.

\section{Proof of Theorem~\ref{thm:fx_de_sim-GFQ}}
\label{app:fx_de_sim-GFQ}

First, we define yet another refinement transformer,
called \emph{fixed-word delayed transformer $\taufxde$},
which is the same as $\taude$,
with the only difference that Spoiler has to reveal the whole input word $w=a_0a_1\cdots$ in advance.
Notice that $\taufxde$, though not efficiently computable in general,
has properties very similar to $\taude$.
In particular, the proof of Theorem~\ref{thm:taude_GFQ} works as it is for the lemma below.

\begin{lemma}\label{lem:taufxde_GFQ}
	If $R$ is a jumping-safe $F$-respecting preorder,
	then, $\taufxde(R)$ is jumping-safe.
\end{lemma}

\appendixstatement{Theorem}{thm:fx_de_sim-GFQ}{
	$\fxdesim$ is good for quotienting.
}
\begin{proof}
	Directly from Lemma~\ref{lem:taufxde_GFQ},
	since $\fxdesim$ is (the transpose of) $\taufxde$ applied to the identity relation.
\end{proof}

\end{document}